\documentclass[a4paper,12pt]{article}
\usepackage[bottom=4.5cm,top=4cm,left=2.2cm,right=2.2cm]{geometry}
\usepackage{graphicx}
\usepackage{subcaption}
\usepackage[shortlabels]{enumitem}

\usepackage{amsthm}
\usepackage{amsmath}
\usepackage{amssymb}
\usepackage{chngcntr}
\usepackage{url}

\usepackage{algorithmicx}
\usepackage{algpseudocode}
\usepackage{algorithm}

\newcommand{\free}{{\mathop{\rm free}}}
\newcommand{\stopp}{{\mathop{\rm stop}}}

\newcommand{\node}{{\mathop{\rm node}}}

\newcommand{\RNum}[1]{\mathop{\rm \uppercase\expandafter{\romannumeral #1\relax}}}

\newtheorem{lemma}{Lemma}[section]

\newtheorem{theorem}{Theorem}[section]

\newtheorem{remark}{Remark}[section]

\newtheorem{definition}{Definition}[section]
\newtheorem{property}{Property}[section]

\newcounter{algcounter}

\begin{document}

\begin{center}
{\Large Approximate Ridesharing of Personal Vehicles Problem} \footnote{A preliminary version of the paper appeared in the Proceeding of COCOA2020 \cite{GLZ20}.}
\vskip 0.1in
Qian-Ping Gu$^a$, Jiajian Leo Liang$^a$ and Guochuan Zhang$^b$

$^a$School of Computing Science, Simon Fraser University, Canada\\
qgu@sfu.ca, leo\_liang@sfu.ca

$^b$College of Computer Science and Technology, Zhejiang University, China\\
zgc@zju.edu.cn
\end{center}

\noindent
{\bf Abstract.}
It is important to find ride matches for individuals who participate in ridesharing quickly, and it is equally important to minimize the number of drivers to serve all individuals and minimize the total travel distance of the vehicles. 
This paper considers the following ridesharing problem: given a set of trips, each trip consists of an individual, a vehicle of the individual and some requirements, select a subset of trips and use the vehicles of selected trips to deliver all individuals to their destinations while satisfying the requirements and achieving some optimization goal.
Requirements of trips are specified by parameters including source, destination, vehicle capacity, preferred paths, detour distance and number of stops a driver is willing to make, and time constraints.
We consider two optimization problems: minimizing the number of selected vehicles and minimizing total travel distance of the vehicles.
We prove that it is NP-hard to approximate both minimization problems within a constant factor if any one of the requirements related to the detour distance, preferred paths, number of stops and time constraints is not satisfied. 
We give $\frac{K+2}{2}$-approximation algorithms for minimizing the number of selected vehicles when the requirement related to the number of stops is not satisfied, where $K$ is the largest capacity of all vehicles.
\vspace{2mm}


\section{Introduction} \label{sec-intro}
The use of shared mobility (carpooling/ridesharing) is becoming popular around the world. With recent advancement in communication technologies, ridesharing in large scale are emerging. Ridesharing services are enabling timely and convenient transportation to many people. The need of such services is increasing as the population grows in urban areas.
At the same time, the number of cars on the road also increases.
According to~\cite{ATST-S}, personal vehicles are the main transportation mode in more than 200 European cities between 2001 and 2011.
In the United States, personal vehicles are also the main transportation mode~\cite{CSS20}. The occupancy rate of personal vehicles in the United States is 1.6 persons per vehicle in 2011~\cite{USDT-S,USDTFHA-S} (and decreased to 1.5 persons per vehicle in 2017~\cite{CSS20}), which can be a major cause for congestion, and the estimated cost of congestion is a round \$121 billion per~\cite{PNAS17-AM}.
Shared mobility (carpooling or ridesharing) is a promising effective way to increase the occupancy rate, which can reduce congestion~\cite{TRTPB-A,TRBM13-F}.
It is estimated that ridesharing to work in Dublin, Ireland, can reduce 12,674 tons of CO$_2$ emissions per year~\cite{TRDTE09-C}, and taxi-ridesharing in Beijing can reduce 120 million liters of gasoline annually~\cite{TKDE15-M}.
A number of systems for ridesharing services are known, such as Uber, Lyft and DiDi.
These systems are also called mobility-on-demand (MoD) systems and can support dynamic ridesharing: ridesharing requests arrive dynamically, and the system provides a service in real-time.
In this paper, we study the static ridesharing problem but our algorithms can be applied to dynamic services. A sequence of dynamic ridesharing requests within a time interval can be viewed as a set of static requests and solved by a static ridesharing algorithm~\cite{ESA15-S}.

Due to the COVID-19 pandemic, traffic volume has decreased significantly in many areas~\cite{IJTST20-D}\footnote{Google COVID-19 Community Mobility Reports, 2021-02. \url{https://google.com/covid19/mobility}}$^{,}$\footnote{Geotab Blog, 2021-02. \url{https://geotab.com/blog/congestion-and-commercial-traffic}}.
The recent study~\cite{IJTST20-D} of pandemic impacts on road traffic found that one of the most effective ways to reduce traffic fuel consumption and emissions is indeed by reducing the number of vehicles on the road, and the authors suggested that the policy makers should encourage ridesharing to reduce the number of vehicles on the road after the pandemic.
During the pandemic, in addition to reducing emissions, ridesharing may be safer than public transit for colleagues who work at the same place; such a type of ridesharing is also easier for contact tracing.
A major lesson from the pandemic is that people should cooperate to protect our living environment, and ridesharing is an effective way to reduce the vehicles on the road, and thus reduce pollution~\cite{IJTST20-D,TRBM13-F}.
In this paper, we focus on this goal.

More specifically, we study two minimization problems in the following ridesharing problem:
Given a set of trips (requests) in a road network, where each trip consists of an individual, a vehicle of the individual and some requirements, select a subset of trips to deliver the individuals of all trips to their destinations by the vehicles of the selected trips satisfying the requirements.
We call the individual of a selected trip a {\em driver} and an individual other than a driver a {\em passenger}. 
The parameters specifying the requirements of a trip include the source and destination of the trip, the vehicle capacity (number of seats to serve passengers), the preferred paths of individual when selected as a driver, the detour distance and number of stops the driver willing to make to serve passengers, and time constraints such as departure/arrival times.
Our optimization goals are to minimize the number of selected vehicles (or equivalently number of drivers) and to minimize the total travel distance of selected vehicles (drivers) to serve all trips.

In general, the ridesharing problem is a generalization of the vehicle routing problem (VRP) and Dial-A-Ride problem (DARP)~\cite{TRBM19-M}, and thus, it is NP-hard.
Mixed Integer Programming (MIP) formulation combined with exact methods or heuristics to solve the MIP formulation have been used to solve the ridesharing problem~\cite{PNAS17-AM,GECCO12-H,VLDBE14-H}.
The MIP based approach is time consuming and not practical for the ridesharing of large scale. Most previous works use (meta)heuristics or study the simplified ridesharing for large instances~\cite{TRBM11-A,CCIE16-J,PNAS14-S,ESA15-S}.
There are two types of optimization goals in ridesharing: {\em operational objectives} and {\em quality-related objectives}~\cite{TRBM19-M}. The operational objectives are to optimize system-wide optimization performances such as maximizing the number of matched/served trips and minimizing the total travel distance/time of all vehicles.
The quality-related objectives are to improve the satisfactions of individuals (drivers/passengers), such as minimizing the waiting time of each individual passenger and maximizing the cost saving of the drivers/passengers.
This is done in first-come first-serve manner for ridesharing service users and is usually achieved by agent-based or decentralized approach by simulating the interaction between passenger ride-requests and driver ride-services (e.g.~\cite{CCGRID18-B,TRCET14-F,TRCET16-N}).
The decentralized approach may be good for providing service to users, but it lacks the ability to achieve operational objectives which optimize system performance.
The centralized approach, using MIP or heuristics, can achieve system-wide optimization goals as a whole (approximately for large instances).
A number of variants and mathematical formulations of the ridesharing problem are derived from DARP and a review on DARP can be found in~\cite{AOR17-M}. Readers may refer to~\cite{TRBM12-A,TRBM13-F,TRBM19-M} for literature surveys and reviews on the ridesharing problem. 

Previous works mainly focus on empirical studies of the ridesharing problem.
Recently, a model for analyzing the relations between the computational complexity of the ridesharing problem and its parameters was introduced by Gu. et al~\cite{GLZ16}; an algorithmic analysis for the simplified ridesharing problem with parameters of source, destination, vehicle capacity, detour distance limit and preferred paths only.
It is shown in~\cite{GLZ16} that if one of the following conditions is not satisfied then both minimizing the number of drivers and minimizing the total travel distance of the drivers are NP-hard: 
\begin{itemize}
\item[(1)] All trips have the same destination or all trips have the same source.
\item[(2)] Detour is not allowed for every trip (zero detour condition).
\item[(3)] Each trip has a unique preferred path (fixed path condition).
\end{itemize}
When all of Conditions (1)-(3) are satisfied, the following exact algorithms are given in~\cite{GLZ19}: $O(M+l^3)$ time dynamic programming algorithms for both minimization problems, where $M$ is the size of road map and $l$ is the number of trips, and an $O(M + l\cdot \log l)$ time exact algorithm for minimizing the number of drivers.

Kuteil and Rawitz~\cite{ESA17-K} studied the maximum carpool matching problem (MCMP), which is closely related to the ridesharing problem.
An instance of MCMP consists of a directed graph $H(V,E)$, where the vertices of $V$ represent the individuals and an arc $(u,v) \in E$ denotes $v$ can serve $u$.
Every $v \in V$ is flexible to be a driver or passenger. The goal of MCMP is to select a set of drivers $S \subseteq V$ to serve all trips of $V$ such that the number of passengers is maximized; MCMP is NP-hard~\cite{ICAM14-H}.
Approximation algorithms are proposed in~\cite{ESA17-K}, and these algorithms can be modified into $\frac{K+2}{2}$-approximation algorithms for minimizing the number of drivers in the ridesharing problem, where $K$ is the largest capacity of all vehicles.

In this paper, we extend the computational complexity analysis of the simplified ridesharing problem to more generalized problems with three additional parameters: the number of stops a driver willing to make to serve passengers, departure time and arrival time of each trip.
Two more conditions related to the three parameters are considered:
\begin{itemize}
\item[(4)] Each driver is willing to stop as many times as its vehicle capacity to pick-up passengers.
\item[(5)] All trips have the same earliest departure time and same latest arrival time.
\end{itemize}
We call Condition (4) the \textit{stop constraint condition} and Condition (5) the \textit{time constraint condition}.
Our results in this paper are:
\begin{enumerate}
\item We prove that it is NP-hard to approximate within a constant factor for each problem of minimizing the number of drivers and minimizing the total travel distance of drivers if stop constraint or time constraint condition is not satisfied.
\item We further show that it is NP-hard to approximate within a constant factor for each problem of minimizing the number of drivers and minimizing the total travel distance of drivers if Condition (2) (zero detour) or Condition (3) (fixed path) is not satisfied.
\item We give two $\frac{K+2}{2}$-approximation algorithms for minimizing the number of drivers when the input instances satisfy Conditions (1-3) and (5), where $K$ is the largest capacity of all vehicles.
For a ridesharing instance of a road network with size $M$ and $l$ trips, our first algorithm, which is a modification of of an approximation algorithm (StarImprove) for the MCMP problem in~\cite{ESA17-K}, runs in $O(M + K \cdot l^3)$ time.
Our second algorithm is more practical and runs in $O(M + l^2)$ time.
\end{enumerate}

\paragraph{Application of ridesharing} In practice, our algorithms may apply to the following described scenario: A ridesharing scenario in regular school commute may be represented by an instance satisfying Conditions (1)-(3) and (5). 
In the morning, many students and staffs go to the same university/college campus (Condition (1), same destination) around the same time (Condition (5)).
Each driver always wants to use a fixed path (usually the fastest route) from home to school (Condition (3)) and does not want to detour (Condition (2)) because time constraints may be tight and traffic can be unpredictable during the peak hours.
Then in the afternoon, staffs and students leave from the same school (Condition (1), same source) around the same time;
from the similar reasons, drivers may want to use a fixed path from school to home and do not want to detour.
Depending on the time constraints, a driver may want to only stop once or twice to pick-up passengers such that the driver's travel duration is not prolonged; on the other hand, if a driver wants to stop many times, Condition (4) is satisfied.
There are studies focus on ridesharing for university commute, such as~\cite{PSBS11-B,Sustain18-D,IJST15-E}.
The work commute scenario is similar to school commute, except the destinations may be scattered around an office area.
In this case and a more dynamic setting for ridesharing, one can apply our algorithms by grouping users together who satisfy (or nearly satisfy) Conditions (1-3) and (5).
This grouping technique has been applied to solving the traveling salesman problem with time windows problem (TSPTW)~\cite{JORS20-B}, which may be useful in minimizing the total travel distance of drivers for the ridesharing problem.

Other studies have shown the possible potential of ridesharing involving autonomous vehicles and (autonomous) taxis~\cite{PNAS17-AM,TRCET18-L}.
It can further benefit the use of autonomous vehicles by computing solution with minimum number of vehicles or minimum total travel distance of vehicles.
Another interesting application is multimodal transportation with ridesharing (integrating public and private transportation). This area of research has gained some attraction recently (e.g.~\cite{TOITS19-H,TRELTR19-M,COR18-S}).
It may be possible to apply our algorithms to this area to satisfy public transportation users demand.

The rest of the paper is organized as follows.
We give in Section \ref{sec-preliminaries} the preliminaries of the paper.
We prove the inapproximability results for stop constraint condition and time constraint condition in Sections \ref{sec-stop-nphard} and \ref{sec-time-nphard}, respectively.
The inapproximability results for Conditions (2) and (3) are given in Section \ref{sec-inapprox-previous-results}.
We modify an approximation algorithm for the MCMP problem into an approximation algorithm for minimizing the number of drivers in Section \ref{sec-MCMP-alg} and present a more practical approximation algorithm for the same minimization problem in Section \ref{sec-new-alg}.
The final section concludes the paper.

\section{Preliminaries} \label{sec-preliminaries}
A (undirected) graph $G$ consists of a set $V(G)$ of vertices and a set $E(G)$ of edges, where each edge $\{u,v\}$ of $E(G)$ is a (unordered) pair of vertices in $V(G)$.
A digraph $H$ consists of a set $V(H)$ of vertices and a set $E(H)$ of arcs, where each arc $(u, v)$ of $E(H)$ is an ordered pair of vertices in $V(H)$.
A graph $G$ (digraph $H$) is weighted if every edge of $G$ (arc of $H$) is assigned a real number as the edge length.
A {\em path} between vertex $v_0$ and vertex $v_k$ in graph $G$ is a sequence $e_1,..,e_k$ of edges, where $e_i=\{v_{i-1},v_i\}\in E(G)$ for $1\leq i\leq k$ and $v_i \neq v_j$ for $i \neq j$ and $0 \leq i, j \leq k$.
A path from vertex $v_0$ to vertex $v_k$ in a digraph $H$ is defined similarly with each $e_i=(v_{i-1},v_i)$ an arc in $H$.
The {\em length} of a path $P$ is the sum of the lengths of edges (arcs) in $P$.
For simplicity, we express a road network by a weighted undirected graph $G(V,E)$ with non-negative edge length: $V(G)$ is the set of locations in the network, an edge $\{u,v\}$ represents the road segment between $u$ and $v$.

In the ridesharing problem, we assume that the individual of every trip can be assigned as a driver or passenger.
In addition to a vehicle and individual, each trip has 
a source, a destination, a capacity of the vehicle, a set of preferred (optional) paths (e.g., shortest paths) to reach the destination, a limit (optional) on the detour distance/time from the preferred path to serve other individuals, a limit (optional) on the number of stops a driver wants to make to pick-up passengers, an earliest departure time, and a latest arrival time.
Each trip in the ridesharing problem is expressed by an integer label $i$ and specified by parameters $(s_i, t_i, n_i, d_i, \mathcal{P}_i, \delta_i, \alpha_i, \beta_i)$, which are defined in Table~\ref{table-para}.
\begin{table}[!ht]
\begin{center}
   \begin{tabular}{| c | l |}
   	\hline
   	\textbf{Parameter} & \textbf{Definition}                                                   \\ \hline
   	$s_i$              & The source (start location) of $i$ (a vertex in $G$)                  \\ 
   	$t_i$              & The destination of $i$ (a vertex in $G$)                              \\ 
   	$n_i$              & The number of seats (capacity) of $i$ available for passengers        \\ 
   	$d_i$              & The detour distance limit $i$ willing to make for offering services   \\ 
   	${\cal P}_i$       & The set of preferred paths of $i$ from $s_i$ to $t_i$ in $G$          \\ 
   	$\delta_i$         & The maximum number of stops $i$ willing to make to pick-up passengers \\ 
   	$\alpha_i$         & The earliest departure time of $i$                                    \\ 
   	$\beta_i$          & The latest arrival time of $i$                                        \\ \hline
   \end{tabular}
\caption{Parameters for a trip $i$.}
\label{table-para}
\end{center}
\vspace{-2mm}
\end{table}

When the individual of trip $i$ delivers (using $i$'s vehicle) the individual of a trip $j$, we say trip $i$ {\em serves} trip $j$ and call $i$ a {\em driver} and $j$ a {\em passenger}.
The \emph{serve relation} between a driver $i$ and a passenger $j$ is defined as follows.
A trip $i$ can serve $i$ itself and can serve a trip $j\neq i$ if $i$ and $j$ can arrive at their destinations by time $\beta_i$ and $\beta_j$ respectively such that $j$ is a passenger of $i$, the detour of $i$ is at most $d_i$, and the number of stops $i$ has to make to serve $j$ is at most $\delta_i$.
When a trip $i$ \emph{can serve} another trip $j$, it means that $i$-$j$ is a feasible assignment of a driver-passenger pair.
We extend this notion to a set $\sigma(i)$ of passenger trips that can be served by a driver $i$ ($i \in \sigma(i)$).
A driver $i$ can serve all trips of $\sigma(i)$ if the total detour of $i$ is at most $d_i$, the number of stops $i$ have to make to pick-up $\sigma(i)$ is at most $\delta_i$, and every $j \in \sigma(i)$ arrives at $t_j$ before $\beta_j$.
At any specific time point, a trip $i$ can serve at most $n_i+1$ trips. If trip $i$ serves some trips after serving some other trips (known as {\em re-take passengers} in previous studies), trip $i$ may serve more than $n_i+1$ trips. In this paper, we study the ridesharing problem in which no re-taking passenger is allowed.
A serve relation is \emph{transitive} if $i$ can serve $j$ and $j$ can serve $k$ imply $i$ can serve $k$.
Let $(G,R)$ be an instance of the ridesharing problem, where $G$ is a road network (weighted graph) and $R=\{1,..,l\}$ is a set of trips. $(S,\sigma)$, where $S\subseteq R$ is a set of trips assigned as drivers and $\sigma$ is a mapping $S\to 2^R$, is a partial solution to $(G,R)$ if
\begin{itemize}
\item for each $i\in S$, $i$ can serve $\sigma(i)$, 
\item for each pair $i,j \in S$ with $i\neq j$, $\sigma(i)\cap \sigma(j)=\emptyset$, and 
\item $\sigma(S)=\cup_{i\in S} \sigma(i)\subseteq R$.
\end{itemize}
When $\sigma(S)=R$, $(S,\sigma)$ is called a solution of $(G,R)$. For a (partial) solution $(S,\sigma)$ we sometimes simply call $S$ a (partial) solution when $\sigma$ is clear from the context or not related to the discussion.

We consider the problem of minimizing $|S|$ (the number of drivers) and the problem of minimizing the total travel distance of the drivers in $S$.
To investigate the relations between the computational complexity and problem parameters, Gu et al.~\cite{GLZ16} introduced the simplified minimization (ridesharing) problems with parameters $(s_i, t_i, n_i, d_i, \mathcal{P}_i)$ only and the following conditions:
\begin{itemize}
\item[(1)] All trips have the same destination or all trips have the same source, that is,
$t_i=D$ for every $i\in R$ or $s_i=\chi$ for every $i\in R$.
\item[(2)] Zero detour: each trip can only serve others on his/her preferred 
path, that is, $d_i=0$ for every $i\in R$.
\item[(3)] Fixed path: ${\cal P}_i$ has a unique preferred path $P_i$.
\end{itemize}
It is shown in~\cite{GLZ16} that if any one of Conditions (1), (2) and (3) is not satisfied, both minimization problems are NP-hard. Polynomial time exact algorithms are given in~\cite{GLZ19} for the simplified minimization problems if all of Conditions (1-3) and transitive serve relation are satisfied.
In this paper, we study more generalized minimization problems with all parameters in Table~\ref{table-para} considered.
To analyze the computational complexity of the more generalized minimization problems, we introduce two more conditions:
\begin{itemize}
\item[(4)] The number of stops each driver is willing to make to pick-up passengers is at least its capacity, that is, $\delta_i \geq n_i$ for every $i \in R$ (stop constraint).
\item[(5)] All trips have the same earliest departure time and same latest arrival time, that is, for every $i \in R$, $\alpha_i = \alpha$ and $\beta_i = \beta$ for some $\alpha < \beta$ (time constraint).
\end{itemize}
The polynomial-time exact algorithms in~\cite{GLZ19} can still apply to any ridesharing instance when all of Conditions (1-5) and transitive serve relation are satisfied.

\section{Inapproximabilities for stop constraint condition} \label{sec-stop-nphard}
We first show the NP-hardness results for the stop constraint condition, that is, when Conditions (1)-(3) and (5) are satisfied but Condition (4) is not.
When Condition (1) is satisfied, we assume all trips have the same destination (since it is symmetric to prove the case that all trips have the same source).
If all trips have distinct sources, one can solve both minimization problems by using the polynomial-time exact algorithms in~\cite{GLZ19}: when Conditions (1-3) are satisfied and each trip has a distinct source $s_i$, each trip is represented by a distinct vertex $i$ in the serve relation graph in~\cite{GLZ19}.
Each time a driver $i$ serves a trip $j$, $i$ must stop at $s_j \neq s_i$ to pick-up $j$. When Condition (4) is not satisfied ($\delta_i < n_i$), $i$ can serve at most $\delta_i$ passengers. Therefore, we can set the capacity $n_i$ to $\min\{n_i,\delta_i\}$ and apply the exact algorithms to solve the minimization problems.
In what follows, we assume trips have arbitrary sources (multiple trips may have a same source).

\subsection{Both minimization problems are NP-hard} \label{sec-stop-nphard-results}
We first prove both minimization problems are NP-hard. These proofs will provide a base for proving the inapproximabilities for both minimization problems.
The NP-hardness proofs are a reduction from the 3-partition problem. The decision problem of 3-partition is that given a set $A = \{a_1, a_2, ..., a_{3r}\}$ of $3r$ positive integers, where $r \geq 2$, $\sum^{3r}_{i=1} a_i = rM$ and $M/4 < a_i < M/2$, whether $A$ can be partitioned into $r$ disjoint subsets $A_1,A_2,....,A_r$ such that each subset has three elements of $A$ and the sum of integers in each subset is $M$. Given a 3-partition instance $A = \{a_1,..., a_{3r}\}$, construct a ridesharing problem instance $(G, R_A)$ as follows (also see Figure~\ref{fig-C1C2C3}).
\begin{itemize}
\item $G$ is a graph with $V(G) = \{D, u_1,..., u_{3r}, v_1,..., v_r\}$ and $E(G)$ having edges $\{u_i, v_1\}$ for $1 \leq i \leq 3r$, edges $\{v_i, v_{i+1}\}$ for $1 \leq i \leq r-1$ and $\{v_r,D\}$. Each edge $\{u,v\}$ has weight of 1, representing the travel distance from $u$ to $v$. It takes $r+1$ units of distance traveling from $u_i$ to $D$ for $1 \leq i \leq 3r$.
\item $R_A = \{1,...,3r+rM\}$ has $3r + rM$ trips. Let $\alpha$ and $\beta$ be valid constants representing time.
	\begin{itemize}
	\item Each trip $i$, $1 \leq i \leq 3r$, has source $s_i = u_i$, destination $t_i = D, n_i = a_i, d_i = 0, \delta_i = 1$, $\alpha_i = \alpha$ and $\beta_i = \beta$. Each trip $i$ has a preferred path $\{u_i, v_1\}, \{v_1,v_2\},...,\{v_r,D\}$ in $G$. 
	\item Each trip $i$, $3r+1 \leq i \leq 3r+rM$, has source $s_i = v_j$, $j = \lceil(i - 3r)/M\rceil$, destination $t_i = D$, $n_i = 0$, $d_i = 0$, $\delta_i = 0$, $\alpha_i = \alpha$, $\beta_i = \beta$ and a unique preferred path $\{v_j, v_{j+1}\},\{v_{j+1},v_{j+2}\},...,\{v_r,D\}$ in $G$.
	\end{itemize}
\end{itemize}

\begin{figure}[htbp]
\centering
\includegraphics[scale=0.68]{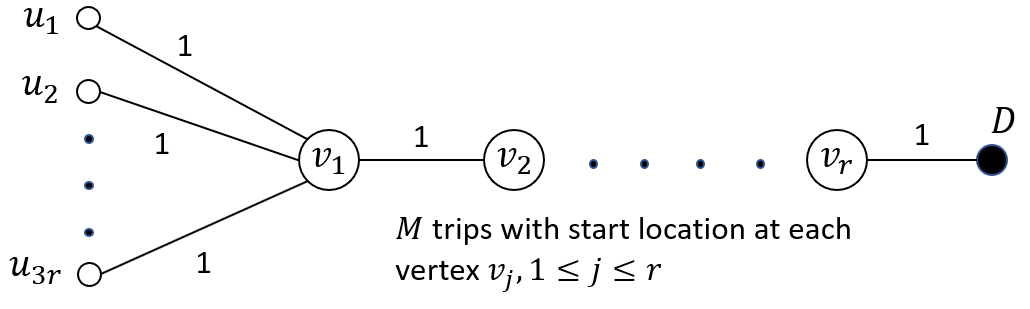}
\caption{Ridesharing instance based on a given 3-partition problem instance.}
\label{fig-C1C2C3}
\end{figure}

\begin{lemma}
Any solution for the instance $(G, R_A)$ has every trip $i$, $1 \leq i \leq 3r$, as a driver and total travel distance at least $3r \cdot (r+1)$.
\label{lemma-C1C2C3}
\end{lemma}

\begin{proof}
Since Condition (2) is satisfied (detour is not allowed), every trip $i$, $1 \leq i \leq 3r$, must be a driver in any solution. A solution with exactly $3r$ drivers has total travel distance $3r \cdot (r+1)$, and any solution with a trip $i$, $3r+1 \leq i \leq 3r+rM$, as a driver has total travel distance greater than $3r \cdot (r+1)$.
\end{proof}

\begin{lemma}
Minimizing the number of drivers in the ridesharing problem is NP-hard when Conditions (1-3) and (5) are satisfied, but Condition (4) is not.
\label{lemma-C1C2C3-size}
\end{lemma}

\begin{proof}
We prove the lemma by showing that an instance $A = \{a_1,..., a_{3r}\}$ of the 3-partition problem has a solution if and only if the ridesharing problem instance $(G,R_A)$ has a solution of $3r$ drivers.

Assume that instance $A$ has a solution $A_1, ..., A_r$ where the sum of elements in each $A_j$ is $M$. For each $A_j = \{a_{j_1},a_{j_2},a_{j_3}\}$, $1 \leq j \leq r$, assign the three trips whose $n_{j_1} = a_{j_1}$, $n_{j_2} = a_{j_2}$ and $n_{j_3} = a_{j_3}$ as drivers to serve the $M$ trips with sources at vertex $v_j$. Hence, we have a solution of $3r$ drivers for $(G,R_A)$.

Assume that $(G,R_A)$ has a solution of $3r$ drivers. By Lemma~\ref{lemma-C1C2C3}, every trip $i$, $1 \leq i \leq 3r$, is a driver in the solution. Then, each trip $j$ for $3r+1 \leq j \leq 3r+rM$ must be a passenger in the solution, total of $rM$ passengers.
Since $\sum_{1 \leq i \leq 3r} a_i = rM$, each driver $i$, $1 \leq i \leq 3r$, serves exactly $n_i = a_i$ passengers. Since $a_i < M/2$ for every $a_i \in A$, at least three drivers are required to serve the $M$ passengers with sources at each vertex $v_j$, $1 \leq j \leq 3r$.
Due to $\delta_i = 1$, each driver $i$, $1 \leq i \leq 3r$, can only serve passengers with the same source. Therefore, the solution of $3r$ drivers has exactly three drivers $j_1, j_2,j_3$ to serve the $M$ passengers with sources at vertex $v_j$, implying $a_{j_1}+a_{j_2}+a_{j_3} = M$. Let  $A_j = \{a_{j_1},a_{j_2},a_{j_3}\}$, $1 \leq j \leq r$, we get a solution for the 3-partition instance.

The size of $(G,R_A)$ is polynomial in $r$. It takes a polynomial time to convert a solution of $(G,R_A)$ to a solution of the 3-partition instance and vice versa. 
\end{proof}

\begin{lemma}
Minimizing the total travel distance of drivers in the ridesharing problem is NP-hard when Conditions (1-3) and (5) are satisfied but Condition (4) is not.
\label{lemma-C1C2C3-distance}
\end{lemma}

\begin{proof}
Let $d_{sum}$ be the sum of travel distances of all trips $i$ with $1 \leq i \leq 3r$. Then the total travel distances of drivers in any solution for $(G,R_A)$ is at least $d_{sum} = 3r(r+1)$ by Lemma~\ref{lemma-C1C2C3}. We show that an instance $A = \{a_1,..., a_{3r}\}$ of the 3-partition problem has a solution if and only if instance $(G,R_A)$ has a solution with travel distance $d_{sum}$.

Assume that the 3-partition instance has a solution. Then there is a solution of $3r$ drivers for $(G,R_A)$ as shown in the proof of Lemma~\ref{lemma-C1C2C3-size}. The total travel distance of this solution is $d_{sum}$.

Assume that $(G,R_A)$ has a solution with total travel distance $d_{sum}$. Trips $i$ with $1 \leq i \leq 3r$ must be drivers. From this, there is a solution for the 3-partition instance as shown in the proof of Lemma~\ref{lemma-C1C2C3-size}.
\end{proof}

\subsection{Inapproximability results} \label{sec-stop-inapprox}
Based on the results in Section~\ref{sec-stop-nphard-results}, we extent our reduction to further show that it is NP-hard to approximate both minimization problems within a constant factor if Condition (4) is not satisfied.
Let $(G,R_A)$ be the ridesharing problem instance constructed based on a given 3-partition instance $A$ as described above for Lemma~\ref{lemma-C1C2C3-size}.
We modify $(G,R_A)$ to get a new ridesharing instance $(G,R')$ as follows.
For every trip $i$, $1 \leq i \leq 3r$, we multiply $n_i$ with $rM$, that is, $n_i = a_i \cdot rM$, where $r$ and $M$ are given in instance $A$. 
There are now $rM^2$ trips with sources at vertex $v_j$ for $1 \leq j \leq r$, and all such trips have the same destination, capacity, detour, stop limit, earlier departure time, latest arrival time, and preferred path as before.
The size of $(G,R')$ is polynomial in $r$ and $M$.
Note that Lemma~\ref{lemma-C1C2C3} holds for $(G,R')$ and $\sum_{i=1}^{3r} n_i = rM \sum_{i=1}^{3r} a_i = (rM)^2$.

\begin{lemma}
Let $(G,R')$ be a ridesharing problem instance constructed above from a 3-partition problem instance $A = \{a_1,\ldots, a_{3r}\}$.
The 3-partition problem instance $A$ has a solution if and only if the ridesharing problem instance $(G,R')$ has a solution $(\sigma, S)$ s.t. $3r \leq |S| < 3r + rM$, where $S$ is the set of drivers.
\label{lemma-stop-inapprox}
\end{lemma}

\begin{proof}
Assume that instance $A$ has a solution $A_1,\ldots,A_r$ where the sum of elements in each $A_j$ is $M$. For each $A_j = \{a_{j_1},a_{j_2},a_{j_3}\}$, $1 \leq j \leq r$, we assign the three trips whose $n_{j_1} = a_{j_1} \cdot rM$, $n_{j_2} = a_{j_2} \cdot rM$ and $n_{j_3} = a_{j_3} \cdot rM$ as drivers to serve the $rM^2$ trips with sources at vertex $v_j$. Hence, we have a solution of $3r$ drivers for $(G,R')$.

Assume that $(G,R')$ has a solution with $3r \leq |S| < 3r + rM$ drivers.
Let $R'(1,3r)$ be the set of trips in $R'$ with labels from $1$ to $3r$.
By Lemma~\ref{lemma-C1C2C3}, every trip $i \in R'(1,3r)$ is a driver in $S$.
Since $a_i < M/2$ for every $a_i \in A$, $n_i < rM \cdot M/2$ for every trip $i \in R'(1,3r)$.
From this, it requires at least three drivers in $R'(1,3r)$ to serve the $rM^2$ trips with sources at each vertex $v_j$, $1 \leq j \leq r$.
For every trip $i \in R'(1,3r)$, $i$ can only serve passengers with the same source due to $\delta_i = 1$.
There are two cases: (1) $|S| = 3r$ and (2) $3r < |S| < 3r + rM$.

(1) It follows from the proof of Lemma~\ref{lemma-C1C2C3-size} that every three drivers $j_1, j_2, j_3$ of the $3r$ drivers serve exactly $rM^2$ passengers with sources at vertex $v_j$.
Then similarly, let $A_j = \{a_{j_1},a_{j_2},a_{j_3}\}$, $1 \leq j \leq r$, we get a solution for the 3-partition problem instance.

(2) For every vertex $v_j$, let $X_j$ be the set of trips with source $v_j$ not served by drivers in $R'(1,3r)$. Then $0 \leq |X_j| < rM$ due to $|S| < 3r + rM$.
For every trip $i \in R'(1,3r)$, $n_i = a_i \cdot rM$ is a multiple of $rM$.
Hence, the sum of capacity for any trips in $R'(1,3r)$ is also a multiple of $rM$, and further, $n_i + n_{i'} = (a_i + a_{i'}) \cdot rM < rM \cdot (M-1)$ for every $i, i' \in R'(1,3r)$ because $a_i<M/2$ and $a_{i'}<M/2$.
From these and $|X_j| < rM$, there are 3 drivers $j_1,j_2,j_3\in R'(1,3r)$ to serve trips with source $v_j$ and $n_{j_1} + n_{j_2} + n_{j_3} \geq rM^2$.
Because $n_{j_1}+n_{j_2}+n_{j_3} \geq rM^2$ for every $1\leq j\leq r$ and $\sum_{1\leq i\leq 3r} n_i = (rM)^2$, $n_{j_1} + n_{j_2} + n_{j_3} = rM^2$ for every $j$.
Thus, we get a solution with $A_j=\{a_{j_1},a_{j_2},a_{j_3}\}$, $1\leq j\leq r$, for the 3-partition problem. 

It takes a polynomial time to convert a solution of $(G,R')$ to a solution of the 3-partition instance and vice versa.
\end{proof}

\begin{theorem}
Let $(G,R')$ be the ridesharing instance stated above based on a 3-partition instance.
Approximating the minimum number of drivers for $(G,R')$ within a constant factor is NP-hard.
This implies that it is NP-hard to approximate the minimum number of drivers within a constant factor for a ridesharing instance when Conditions (1-3) and (5) are satisfied and Condition (4) is not.
\label{theorem-stop-inapprox}
\end{theorem}

\begin{proof}
Assume that there is a polynomial time $c$-approximation algorithm $C$ for instance $(G,R')$ for any constant $c > 1$.
This means that $C$ will output a solution $(\sigma_C, S_C)$ for $(G,R')$ such that $OPT(R') \leq |S_C| \leq c \cdot OPT(R')$, where $OPT(R')$ is the minimum number of drivers for $(G,R')$.
When the 3-partition instance is a ``No'' instance, the optimal value for $(G,R')$ is $OPT(R') \geq 3r + rM$ by Lemma~\ref{lemma-stop-inapprox}. Hence, algorithm $C$ must output a value $|S_C| \geq 3r + rM$.
When the 3-partition instance is a ``Yes'' instance, the optimal value for $(G,R')$ is $OPT(R') = 3r$.
For any constant $c > 1$, taking $M$ such that $c < M/3 + 1$.
The output $|S_C|$ from algorithm $C$ on $(G,R')$ is $3r \leq |S_C| \leq 3rc < 3r + rM$ for a 3-partition ``Yes'' instance.
Therefore, by running the $c$-approximation algorithm $C$ on any ridesharing instance $(G,R')$ and checking the output value $|S_C|$ of $C$, we can answer the 3-partition problem in polynomial time, which contradicts that the 3-partition problem is NP-hard unless $P=NP$. 
\end{proof}

\begin{theorem}
It is NP-hard to approximate the total travel distance of drivers within any constant factor for a ridesharing instance when Conditions (1-3) and (5) are satisfied and Condition (4) is not.
\label{theorem-stop-distance-inapprox}
\end{theorem}

\begin{proof}
Let $(G,R')$ be the ridesharing problem instance used in Theorem~\ref{theorem-stop-inapprox}, based on a given 3-partition instance $A = \{a_1,..., a_{3r}\}$.
Let $d(S)$ be the sum of travel distances for a set $S$ of drivers.
Let $R'(1,3r)$ be the set of trips in $R'$ with labels from $1$ to $3r$.
By Lemma~\ref{lemma-C1C2C3}, all of $R'(1,3r)$ must be drivers in any solution for $(G,R')$ and $d(R'(1,3r)) = 3r(r+1)$.
Assume that there is a polynomial time $c$-approximation algorithm $C$ for the ridesharing problem $(G,R')$ for any constant $c > 1$.
This means that $C$ will output a solution $(\sigma_C, S_C)$ for $(G,R')$ such that $OPT(R') \leq d(S_C) \leq c \cdot OPT(R')$, where $OPT(R')$ is the minimum total travel distance of drivers for $(G,R')$.
By Lemma~\ref{lemma-stop-inapprox}, when the 3-partition instance is a ``No'' instance, the number of drivers in any solution for $(G,R')$ is at least $3r + rM$. All $rM$ trips (of the $3r + rM$) can have source at vertex $v_r$, so $d(S_C) \geq 3r(r+1) + rM$.
When the 3-partition instance is a ``Yes'' instance, the optimal value for $(G,R')$ is $OPT(R') = 3r(r+1)$.
For any constant $c > 1$, taking $M$ and $r$ such that $c < \frac{M}{3(r+1)} + 1$.
The output $d(S_C)$ from algorithm $C$ on $(G,R')$ is $3r(r+1) \leq d(S_C) \leq 3r(r+1)c < 3r(r+1) + rM$ for a 3-partition ``Yes'' instance.
Therefore, by running the $c$-approximation algorithm $C$ on any ridesharing instance $(G,R')$ and checking the output value $d(S_C)$ of $C$, we can answer the 3-partition problem in polynomial time, which contradicts that the 3-partition problem is NP-hard unless $P=NP$. 
\end{proof}

\section{Inapproximabilities for time constraint condition} \label{sec-time-nphard}
Assume that Conditions (1-4) are satisfied but Condition (5) is not, that is, trips can have arbitrary departure time and arrival time.
Recall that we assume all trips have the same destination when Condition (1) is satisfied (the same reduction with simple modifications can also be applied to all trips have the same source).

\subsection{NP-hardness results}
We first show that both minimization problems are NP-hard, and these proofs serve as part of the inapproximability proofs. The NP-hardness proofs are a reduction from 3-partition problem, which is similar to the one used in Lemma~\ref{lemma-C1C2C3-size}.
Given a 3-partition minimization problem instance, construct a ridesharing instance $(G, R_A)$ with $G$ shown in Figure~\ref{fig-C1C2C3}. The only differences are the values of $\alpha_i$, $\beta_i$ and $\delta_i$.
\begin{itemize}
	\item For trips $i$, $1 \leq i \leq 3r$, source $s_i = u_i$, destination $t_i = D$, $n_i = a_i$, $d_i = 0$, $\delta_i = n_i$, $\alpha_i = 0$, $\beta_i = 2r$. Each trip $i$ has a preferred path $\{u_i, v_1\}, \{v_1,v_2\},...,\{v_r,D\}$ in $G$. 
	\item For trips $i$, $3r+1 \leq i \leq 3r+rM$, source $s_i = v_j$, $j = \lceil(i - 3r)/M\rceil$, destination $t_i = D$, $n_i = 0$, $d_i = 0$, $\delta_i = 0$. Each trip $i$ has a unique preferred path $\{v_j, v_{j+1}\},\{v_{j+1},v_{j+2}\},\ldots,\{v_r,D\}$ in $G$, $\alpha_i = r$ and $\beta_i = 2r - j + 1$, where $j = \lceil(i - 3r)/M\rceil$.
\end{itemize}
Note that every trip $i \in R_A$ has the same travel distance from $s_i$ to $t_i$ as previous construction in Section~\ref{sec-stop-nphard}.
Since they have the same construction, Lemma~\ref{lemma-C1C2C3} also holds for this ridesharing instance $(G, R_A)$.

\begin{lemma}
In any solution for the instance $(G, R_A)$, all trips served by a driver $i \in R_A$ (other than $i$ itself), must have the same source $v_j$, for some $j \in [1,\ldots, 3r]$.
\label{lemma-C1C2C3-serve}
\end{lemma}

\begin{proof}
By Lemma~\ref{lemma-C1C2C3}, every trip $i$, $1 \leq i \leq 3r$, is a driver in any solution. Thus, only trips with source $v_j$, $1 \leq j \leq 3r$, can be passengers. 
Let $j$ be a trip with source $v_j$. The travel time from $v_j$ to $D$ is $r-j+1$. Since $\beta_{j} = 2r - j +1$, $j$ must be picked-up no later than time $r$. Otherwise, $j$ cannot arrive at $t_{j} = D$ by time $\beta_{j}$.
From this and the fact that $\alpha_{j} = r$, $j$ must be picked-up at time $r$ exactly.
Suppose that driver $i$ serves trip $j$. The travel time from $s_i$ to $s_{j} = v_j$ is $j \leq r$. $i$ can arrive at $D$ (after delivering $j$) no later than time $2r = \beta_i$.

Let $j_1$ and $j_2$ be two trips with $s_{j_1} = v_{j_1}$, $s_{j_2} = v_{j_2}$ and $j_1 < j_2$.
Then any trip $i$ with $1\leq i\leq 3r$ can serve only one of $j_1$ and $j_2$ due to the following reasons.
Suppose $i$ picks-up $j_1$ first.
By the time $i$ reaches $v_{j_2}$ after picking-up $j_1$, it will pass time $r$, and from above, $i$ can no longer serve $j_2$. Otherwise, $j_2$ will not be arrive on time.
Suppose $i$ picks-up $j_2$ first.
When $i$ reaches $v_{j_1}$ by going back, it will also pass time $r$. Hence, $i$ cannot serve $j_1$ in this case.
Therefore, if $i$ decides to serve a trip $j$ with source $v_j$, the only other trips $i$ can serve must also have source $v_j$. 
\end{proof}
Lemma~\ref{lemma-C1C2C3-serve} actual implies that every driver $i$ ($1 \leq i \leq 3r$) in any solution for $(G, R_A)$ will only make at most one stop, effectively making $\delta_i = 1$.

\begin{lemma}
Minimizing the number of drivers in the ridesharing problem is NP-hard when Conditions (1-4) are satisfied but Condition (5) is not.
\label{lemma-C1C2C3-time-size}
\end{lemma}

\begin{proof}
Assume that the 3-partition instance has a solution $A_1, ..., A_r$ where the sum of elements in each $A_j$ is $M$. For each $A_j = \{a_{j_1},a_{j_2},a_{j_3}\}$, $1 \leq j \leq r$, we assign the three trips whose $n_{j_1} = a_{j_1}$, $n_{j_2} = a_{j_2}$ and $n_{j_3} = a_{j_3}$ as drivers to serve the $M$ trips with sources at vertex $v_j$. Hence, we have a solution of $3r$ drivers for $(G,R_A)$.

Assume that $(G,R_A)$ has a solution of $3r$ drivers. By Lemma~\ref{lemma-C1C2C3}, every trip $i$, $1 \leq i \leq 3r$, is a driver in the solution. Similarly, each driver $i$ serves exactly $n_i = a_i$ passengers, and at least three drivers are required to serve the $M$ passengers with sources at each vertex $v_j$, $1 \leq j \leq 3r$.
By Lemma~\ref{lemma-C1C2C3-serve}, each driver $i$, $1 \leq i \leq 3r$, can only serve passengers with the same source. Therefore, the solution of $3r$ drivers has exactly three drivers $j_1, j_2,j_3$ to serve the $M$ passengers with sources at the vertex $v_j$, implying $a_{j_1}+a_{j_2}+a_{j_3} = M$. Let  $A_j = \{a_{j_1},a_{j_2},a_{j_3}\}$, $1 \leq j \leq r$, we get a solution for the 3-partition problem instance.

The size of $(G,R_A)$ is polynomial in $r$. It takes a polynomial time to convert a solution of $(G,R_A)$ to a solution of the 3-partition instance and vice versa. 
\end{proof}

With a similar argument to that of Lemma~\ref{lemma-C1C2C3-distance}, we have the following lemma.
\begin{lemma}
Minimizing the total travel distance of drivers in the ridesharing problem is NP-hard when Conditions (1-4) are satisfied but Condition (5) is not.
\label{lemma-C1C2C3-time-distance}
\end{lemma}

\subsection{Inapproximability results}
It is NP-hard to approximate the minimum number of drivers and total travel distance of drivers within a constant factor for the ridesharing problem when Conditions (1-4) are satisfied but condition (5) is not.
The proofs are identical to the inapproximability proof of Theorem~\ref{theorem-stop-inapprox} and Theorem~\ref{theorem-stop-distance-inapprox} respectively for each minimization problem.
Let $(G,R_A)$ be the ridesharing problem instance constructed based on a given 3-partition instance as described in Section~\ref{sec-time-nphard}.
Construct a ridesharing instance $(G,R')$ from $(G,R_A)$ as described in Section~\ref{sec-stop-inapprox}.
Then Lemma~\ref{lemma-stop-inapprox} and Lemma~\ref{lemma-C1C2C3-serve} can be applied to $(G,R')$.
From this, the analysis of Theorem~\ref{theorem-stop-inapprox} and Theorem~\ref{theorem-stop-distance-inapprox} can be applied to $(G,R')$, and we have the following theorems.

\begin{theorem}
It is NP-hard to approximate the minimum number of drivers within any constant factor for a ridesharing instance satisfying Conditions (1-4) but not Condition (5).
\label{lemma-time-approx}
\end{theorem}

\begin{theorem}
It is NP-hard to approximate the minimum total travel distance of drivers within any constant factor for a ridesharing instance satisfying Conditions (1-4) but not Condition (5).
\label{lemma-time-distnace-approx}
\end{theorem}

\section{Inapproximabilities for Conditions (2) and (3)} \label{sec-inapprox-previous-results}
It is proved in~\cite{GLZ16} that the ridesharing minimization problems are NP-hard when all conditions except Condition (2) or Condition (3) are satisfied.
In this section, we show that for each case, it is NP-hard to approximate both minimization problems within any constant factor.
The proof uses a similar method as described in Section~\ref{sec-stop-inapprox}.

\subsection{Inapproximability results for non-zero detour}
Recall that the NP-hardness proof for this case is a reduction from the 3-partition problem too~\cite{GLZ16}.
For completeness, we show the construction of the slightly modified $(G, R_A)$ and the inapproximability proof.
Given a 3-partition instance $A = \{a_1,..., a_{3r}\}$, the ridesharing instance $(G, R_A)$ is constructed as follows (see Figure~(\ref{fig-previousResults}a)):
\begin{itemize}
\item $G$ is a graph with $V(G) = \{I, D, u_1,\ldots,u_r, v_1,\ldots,v_{3r}\}$ and $E(G)$ having edges $\{u_i, I\}$ of weight $rM$ for $1 \leq i \leq r$, edges $\{v_i, I\}$ of weight $a_i$ for $1 \leq i \leq 3r$ and edge $\{I,D\}$ of weight $rM$.
\item $R_A = \{1,...,r+3rrM\}$ has $r + 3rrM$ trips. Let $\alpha$ and $\beta$ be valid constants representing time.
	\begin{itemize}
	\item Each trip $i$, $1 \leq i \leq r$, has source $s_i = u_i$, destination $t_i = D, n_i = 3rM, d_i = 2M, \delta_i = n_i$, $\alpha_i = \alpha$ and $\beta_i = \beta$. Each trip $i$ has a preferred path $\{u_i, I\}, \{I,D\}$ in $G$.
	\item Each trip $i$, $r+1 \leq i \leq r+3rrM$, has source $s_i = v_j$, $j = \lceil(i - r)/rM\rceil$, destination $t_i = D$, $n_i = 0$, $d_i = 0$, $\delta_i = 0$, $\alpha_i = \alpha$, $\beta_i = \beta$ and a unique preferred path $\{v_j, I\},\{I,D\}$ in $G$.
	\end{itemize}
\end{itemize}
\begin{figure}[htbp]
\centering
\includegraphics[scale=0.6]{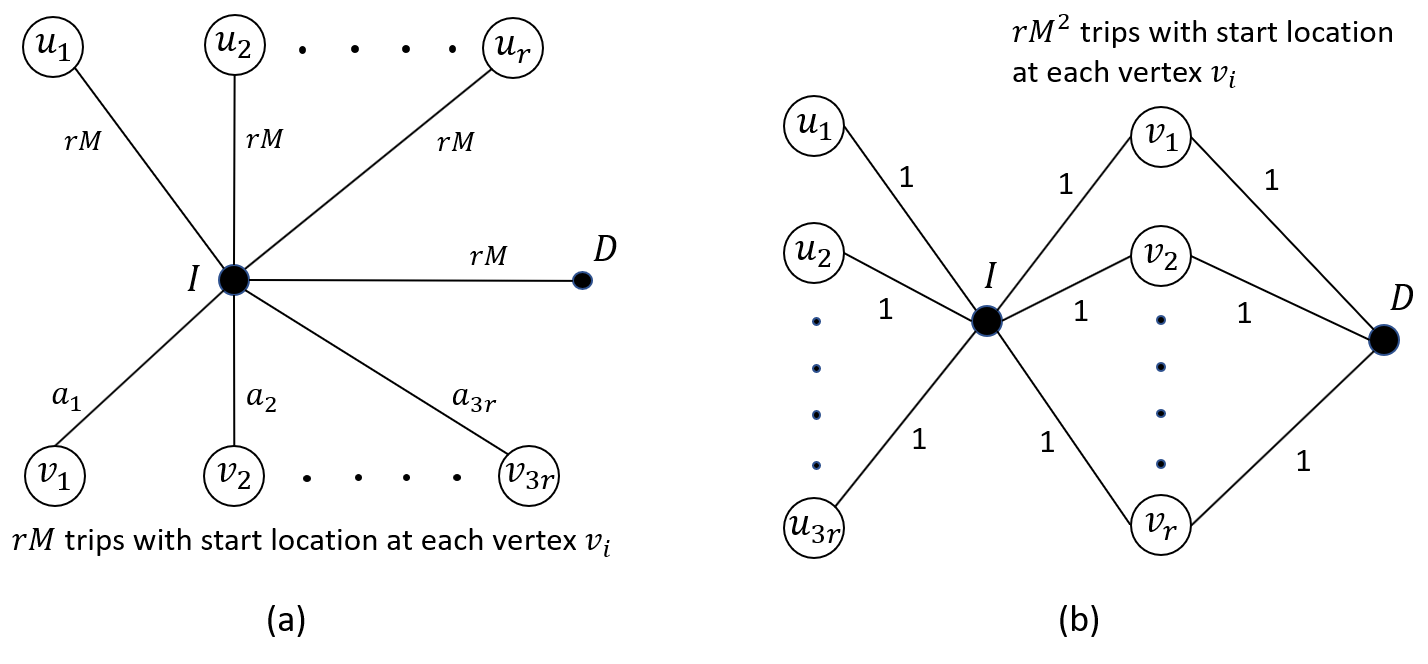}
\caption{Ridesharing instance based on a given 3-partition problem instance: (a) Conditions (1) and (3-5) are satisfied but not Condition (2); (b) Conditions (1-2) and (4-5) are satisfied but not Condition (3).}
\label{fig-previousResults}
\end{figure}
First we re-state Lemma 3.1 in~\cite{GLZ16} as the following lemma.
\begin{lemma}
Any solution for the instance $(G, RA)$ has every trip $i$ with $1 \leq i \leq r$ as a driver and total travel distance at least $2rM(r + 1)$.
\label{lemma-detour}
\end{lemma}

\begin{proof}
Any solution for the instance must have every trip $i$ with $1 \leq i \leq r$ as a driver because of the detour distance limit and the travel distance between the trips $1 \leq i \leq r$.
Let $S = \{i \mid 1 \leq i \leq r\}$ be the set of driver. The total travel distance of $S$ is at least $2rM$.
For each trip $j$ with source $s_j = v_k$ ($1 \leq k \leq 3r$), the total travel distance of drivers in $S$ and trip $j$ is at least $2rrM + 2a_k$ if $j$ is served by a driver in $S$, otherwise is at least $2rrM + a_k + rM$.
Since $a_k < rM$, the minimum total travel distance of any solution is to have every $j$ with source at $v_k$, $1 \leq k \leq 3r$, as a passenger served by $S$ with total distance $2rrM + \sum_{1 \leq k \leq 3r} 2a_k = 2rrM + 2rM = 2rM(r+1)$.
\end{proof}

\begin{lemma}
Let $(G,R_A)$ be a ridesharing problem instance constructed above from a 3-partition problem instance $A = \{a_1,\ldots, a_{3r}\}$.
The 3-partition problem instance $A$ has a solution if and only if the ridesharing problem instance $(G,R_A)$ has a solution $(\sigma, S)$ s.t. $r \leq |S| < r + rM$, where $S$ is the set of drivers.
\label{lemma-detour-inapprox}
\end{lemma}

\begin{proof}
Assume that instance $A$ has a solution $A_1,\ldots,A_r$ where the sum of elements in each $A_i$ is $M$. For each $A_i = \{a_{i_1},a_{i_2},a_{i_3}\}$, $1 \leq i \leq r$, we say trips with sources at vertex $v_k$ ($1 \leq k \leq 3r$) \textit{correspond to} $A_i$ if the edge $\{v_k,I\}$ has weight $a_{i_1},a_{i_2}$, or $a_{i_3}$.
By the definition of 3-partition instance, one can uniquely identify the corresponding trips of $A_i$.
Then for each $A_i$, there are exactly $3rM$ corresponding trips with sources at three different vertices $v_k$.
Recall that every trip $i$ with source at $u_i$ has detour distance limit $d_i = 2M$ and capacity $n_i = 3rM$.
We assign a trip $i$ with source at $u_i$ as a driver to serve the corresponding trips of $A_i$ for $1 \leq i \leq r$.
It requires exactly $2M$ detour distance for $i$ to serve the $3rM$ corresponding trips of $A_i$.
Hence, we have a solution of $r$ drivers for $(G,R_A)$.

Assume that $(G,R_A)$ has a solution with $r \leq |S| < r + rM$ drivers.
Let $R(1,r)$ be the set of trips in $R_A$ with labels from $1$ to $r$.
By Lemma~\ref{lemma-detour}, every trip $i \in R(1,r)$ is a driver in $S$ (trips with source at $u_i$ are drivers).
For every vertex $v_k$, let $X_k$ be the set of trips with source at $v_k$ not served by drivers in $R(1,r)$. Then $0 \leq |\cup_{1\leq k\leq 3r} X_k| < rM$ due to $|S| < r + rM$.
From this and there are $rM$ trips with source at each vertex $v_k$, every driver in $R(1,r)$ must detour to some vertex $v_k$ ($1 \leq k \leq 3r$) to pick-up some passengers. In other words, every vertex $v_k$ for $1 \leq k \leq 3r$ must have been visited by at least one driver in $R(1,r)$.
Assume some driver $i \in R(1,r)$ has detour distance less than $2M$ ($i$ detours to at least 1 vertex and at most 3 vertices because $M/4 < a_k <M/2$ for $1 \leq k \leq 3r$).
Then from the fact that the sum of elements in $A$ is $rM$ ($\sum_{1 \leq k \leq 3r} a_k = rM$), some driver $i'$ must has detour distance greater than $2M$ so that all vertices can be visited.
This is a contradiction to $d_{i'} = 2M$.
Hence, the detour distance of every driver in $R(1,r)$ is exactly $2M$.
For each driver $i$, $1 \leq i \leq r$, let $A_i$ be the subset of the three integers of $A$ corresponding to the detour distance traveled by $i$ (one way). Then $A_1,\ldots, A_r$ is a solution for the 3-partition problem instance.

It takes a polynomial time to convert a solution of $(G,R_A)$ to a solution of the 3-partition instance and vice versa.
\end{proof}

With Lemma~\ref{lemma-detour-inapprox}, the analysis of Theorem~\ref{theorem-stop-inapprox} can be applied to $(G,R_A)$, and we have the following theorem.
\begin{theorem}
It is NP-hard to approximate the minimum number of drivers within any constant factor for a ridesharing instance satisfying Conditions (1) and (3-5) but not Condition (2).
\label{lemma-RS-C1C3-approx}
\end{theorem}

\begin{theorem}
It is NP-hard to approximate the minimum total travel distance of drivers within any constant factor for a ridesharing instance satisfying Conditions (1) and (3-5) but not Condition (2).
\label{theorem-C1C3-detour-distance}
\end{theorem}

\begin{proof}
With a similar argument to that of Theorem~\ref{theorem-stop-distance-inapprox}, we have the theorem.
\end{proof}

\subsection{Inapproximability results for no fixed preferred path}
The NP-hardness proof for this case is also a reduction from the 3-partition problem~\cite{GLZ16}.
Given a 3-partition instance $A = \{a_1,..., a_{3r}\}$, the ridesharing instance $(G, R_A)$ is constructed as follows (see Figure~(\ref{fig-previousResults}b)):
\begin{itemize}
\item $G$ is a graph with $V(G) = \{D, u_1,..., u_{3r}, v_1,..., v_r\}$ and $E(G)$ having edges $\{u_i, I\}$ for $1 \leq i \leq 3r$, edges $\{I, v_i\}$ for $1 \leq i \leq r$ and edges $\{v_i,D\}$ for $1 \leq i \leq r$.
Every edge in $E(G)$ has weight of 1.
\item $R_A = \{1,...,3r+(rM)^2\}$ has $3r + (rM)^2$ trips. Let $\alpha$ and $\beta$ be valid constants representing time.
	\begin{itemize}
	\item Each trip $i$, $1 \leq i \leq 3r$, has source $s_i = u_i$, destination $t_i = D, n_i = a_i\cdot rM, d_i = 0, \delta_i = n_i$, $\alpha_i = \alpha$ and $\beta_i = \beta$. Each trip $i$ has $r$ preferred paths $\{u_i, I\}, \{I,v_i\}, \{v_i,D\}$ in $G$, $1 \leq i \leq r$.
	\item Each trip $i$, $3r+1 \leq i \leq (rM)^2$, has source $s_i = v_j$, $j = \lceil(i - 3r)/rM^2\rceil$, destination $t_i = D$, $n_i = 0$, $d_i = 0$, $\delta_i = 0$, $\alpha_i = \alpha$, $\beta_i = \beta$ and a unique preferred path $\{v_j,D\}$ in $G$.
	\end{itemize}
\end{itemize}
Lemma 3.2 in~\cite{GLZ16} also holds for the instance $(G, R_A)$, which is stated as the following lemma.
\begin{lemma}
Any solution for the instance $(G, RA)$ has every trip $i$ with $1 \leq i \leq r$ as a driver and total travel distance at least $9r$.
\label{lemma-unique-path}
\end{lemma}

\begin{lemma}
Let $(G,R_A)$ be a ridesharing problem instance constructed above from a 3-partition problem instance $A = \{a_1,\ldots, a_{3r}\}$.
The 3-partition problem instance $A$ has a solution if and only if the ridesharing problem instance $(G,R_A)$ has a solution $(\sigma, S)$ s.t. $3r \leq |S| < 3r + rM$, where $S$ is the set of drivers.
\label{lemma-path-inapprox}
\end{lemma}

\begin{proof}
With Lemma~\ref{lemma-unique-path} and Theorem~3.3 in~\cite{GLZ16}, a similar analysis of Lemma~\ref{lemma-stop-inapprox} can be applied to this lemma.
\end{proof}

From Lemma~\ref{lemma-path-inapprox}, the analysis of Theorem~\ref{theorem-stop-inapprox} and Theorem~\ref{theorem-stop-distance-inapprox} can be applied to $(G,R_A)$, and we have the following theorems.

\begin{theorem}
It is NP-hard to approximate the minimum number of drivers within any constant factor for a ridesharing instance satisfying Conditions (1-2) and (3-4) but not Condition (5).
\label{lemma-path-approx}
\end{theorem}
\vspace{-6mm}

\begin{theorem}
It is NP-hard to approximate the minimum total travel distance of drivers within any constant factor for a ridesharing instance satisfying Conditions (1-2) and (3-4) but not Condition (5).
\label{lemma-path-distnace-approx}
\end{theorem}

\section{Approximation algorithms based on MCMP} \label{sec-MCMP-alg}
For short, we call the ridesharing problem with all conditions satisfied except Condition (4) as \emph{ridesharing problem with stop constraint}.
Let $K = \max_{i \in R} {n_i}$ be the largest capacity of all vehicles.
Kutiel and Rawitz~\cite{ESA17-K} proposed two $\frac{1}{2}$-approximation algorithms for the maximum carpool matching problem.
We show in this section that the algorithms in~\cite{ESA17-K} can be modified to $\frac{K+2}{2}$-approximation algorithms for minimizing the number of drivers in the ridesharing problem with stop constraint. Then in the next section, we propose a more practical $\frac{K+2}{2}$-approximation algorithm for the minimization problem.

An instance of the maximum carpool matching problem (MCMP) consists of a directed graph $H(V,E)$, a capacity function $c : V \rightarrow \mathbb{N}$, and a weight function $w : E \rightarrow \mathbb{R^+}$, where the vertices of $V$ represent the individuals and an arc $(u,v) \in E$ implies $v$ can serve $u$.
We are only interested in the unweighted case, that is, $w(u,v) = 1 $ for every $(u,v) \in E$.
Every $v \in V$ can be assigned as a driver or passenger. The goal of MCMP is to find a set of drivers $S \subseteq V$ to serve all $V$ such that the number of passengers is maximized.
A solution to MCMP is a set $\cal{S}$ of vertex-disjoint stars in $H$.
Let $S_v$ be a star in $\cal{S}$ rooted at center vertex $v$, and leaves of $S_v$ is denoted by $P_v = V(S_v) \setminus \{v\}$.
For each star $S_v \in \cal{S}$, vertex $v$ has out-degree of 0 and every leave in $P_v$ has only one out-edge towards $v$.
The center vertex of each star $S_v$ is assigned as a driver and the leaves are assigned as passengers.
The set of edges in $\cal{S}$ is called a matching $M$. An edge in $M$ is called a \emph{matched} edge.
Notice that $|M|$ equals to the number of passengers.
For an arc $e = (u,v)$ in $H$, vertices $u$ and $v$ are said to be \emph{incident to} $e$.
For a matching $M$ and a set $V' \subseteq V$ of vertices, let $M(V')$ be the set of edges in $M$ incident to $V'$.
The \emph{in-neighbors} of a vertex $v$ is defined as $N^{in}(v) = \{u \mid (u,v) \in E\}$, and the set of arcs entering $v$ is defined as \emph{in-arcs} $E^{in}(v) = \{(u,v) \mid (u,v) \in E\}$.
Table~\ref{table-def-mcmp} lists the basic notation and definition for this section.
\begin{table}[!ht]
\begin{center}
   \begin{tabular}{| l | l |}
   	\hline
   	\textbf{Notation} & \textbf{Definition}                                                    \\ \hline
    $\cal{S}$         & A set of vertex-disjoint stars in $H$ (solution to MCMP)               \\ 
    $S_v$ and $P_v$   & A Star $S_v$ rooted at center vertex $v$							   \\
    $P_v$             & $P_v = V(S_v) \setminus \{v\}$, the set of leaves of star $S_v$        \\ 
  	$c(v)$            & Capacity of vertex $v$ (equivalent to $n_v$ in Table~\ref{table-para}) \\  
    Matching $M$      & The set of edges in $\cal{S}$, namely $E(S)$                		   \\  
   	$M(V')$           & The set of edges in $M$ incident to a set $V'$ of vertices             \\ 
   	$N^{in}(v)$       & The set of \emph{in-neighbors} of $v$, $N^{in}(v) = \{u \mid (u,v) \in E\}$  \\ 
    $E^{in}(v)$       & The set of arcs entering $v$, \emph{in-arcs} $E^{in}(v) = \{(u,v) \mid (u,v) \in E\}$   \\
   	$\delta P_v$      & The number of stops required for $v$ to pick-up $P_v$                  \\ \hline
   \end{tabular}
\caption{Common notation and definition used in this section.}
\label{table-def-mcmp}
\end{center}
\vspace{-1mm}
\end{table}

Two approximation algorithms (StarImprove and EdgeSwap) are presented in~\cite{ESA17-K}; both can achieve $\frac{1}{2}$-approximation ratio, that is, the number of passengers found by the algorithm is at least half of that for the optimal solution.

\paragraph{\textbf{EdgeSwap}}
The EdgeSwap algorithm requires the input instance to have a bounded degree graph (or the largest capacity $K$ is bounded by a constant) to have a polynomial running time.
The idea of EdgeSwap is to swap $i$ matched edges in $M$ with $i+1$ edges in $E \setminus M$ for $1 \leq i \leq k$ and $k$ is a constant integer. The running time of EdgeSwap is in the order of $O(|E|^{2k+1})$.
EdgeSwap can directly apply to the minimization problem to achieve $\frac{K+2}{2}$-approximation ratio in $O(l^{2K})$ time, which may not be practical even if $K$ is a small constant.

\paragraph{\textbf{StarImprove}}
Let $(H(V, E), c, w)$ be an instance of MCMP.
Let $\cal{S}$ be the current set of stars found by StarImprove and $M$ be the set of matched edges.
The idea of the StarImprove algorithm is to iteratively check in a \emph{for-loop} for every vertex $v \in V(G)$:
\begin{itemize}
\item check if there exists a star $S_v$ with $E(S_v) \cap M = \emptyset$ such that the resulting set of stars $\cal{S}$ $\setminus M(V(S_v)) \cup S_v$ gives a larger matching.
\end{itemize}
Such a star $S_v$ is called an \textit{improvement} and $|P_v| \leq c(v)$.
Given a ridesharing instance $(G, R)$ satisfying all conditions, except Condition (4). The StarImprove algorithm cannot apply to $(G, R)$ directly because the algorithm assumes a driver $v$ can serve any combination of passengers corresponding to vertices adjacent to $v$ up to $c(v)$.
This is not the case for $(G, R)$ in general. For example, suppose $v$ can serve $u_1$ and $u_2$ with $n_v = 2$ and $\delta_v = 1$.
The StarImprove assigns $v$ as a driver to serve both $u_1$ and $u_2$. However, if $u_1$ and $u_2$ have different sources ($s_v \neq s_{u_1} \neq s_{u_2}$), this assignment is not valid for $(G, R)$.
Hence, we need to modify StarImprove for computing a star.
For a vertex $v$ and star $S_v$, let $N^{in}_{\mbox{-}M}(v) = \{i \mid i \in N^{in} \setminus V(M)\}$ and $\delta P_v$ be the number of stops required for $v$ to pick-up $P_v$.
Suppose the in-neighbors $N^{in}_{\mbox{-}M}(v)$ are partitioned into $g_1(v),\ldots,g_m(v)$ groups such that trips with same source are grouped together.
When stop constraint is considered, finding a star $S_v$ with maximum $|P_v|$ is similar to solving a fractional knapsack instance using a greedy approach as shown in Figure~\ref{fig-alg-star}.
The idea is, in each iteration, to select the largest group of in-neighbors $N^{in}_{\mbox{-}M}(v)$ until the capacity $c(v)$ is reached.
\stepcounter{algcounter}
\begin{figure}[htb]
\small
\textbf{Algorithm~\arabic{algcounter}} Greedy algorithm
\begin{algorithmic}[1]
\\ $P_v = \emptyset$; $c=c(v)$; $\delta P_v = 0$;
\If {$\exists$ a group $g_j(v)$ s.t. $s_u = s_v$ for any $u \in g_j(v)$}
    \State Select $g_i(v) = \max_{1\leq i \leq m : s_u = s_v, u \in g_i(v)} \{|g_i(v) \setminus P_v|\}$;
    \State Let $g'_j(v) \subseteq g_j(v)$ be a maximum subset of $g_j(v)$ such that $|g'_j(v)| \leq c$.
    \State $P_v = P_v \cup g'_j(v)$; $c = c - |g'_j(v)|$;
\EndIf
\While {$c > 0$ and $\delta P_v < \delta_v$}
    \State Select $g_i(v) = \max_{1\leq i \leq m} \{|g_i(v) \setminus P_v|\}$;
    \State Let $g'_i(v) \subseteq g_i(v)$ be a maximum subset of $g_i(v)$ such that $|g'_i(v)| \leq c$.
    \State $P_v = P_v \cup g'_i(v)$; $c = c - |g'_i(v)|$; $\delta P_v = \delta P_v + 1$;
\EndWhile
\State \Return the star $S_v$ induced by $P_v \cup \{v\}$;
\end{algorithmic}
\caption{Greedy algorithm for computing $S_v$.}
\label{fig-alg-star}
\end{figure}

\begin{lemma}
Let $v$ be the trip being processed and $S_v$ be the star found by Algorithm~\arabic{algcounter} w.r.t. current matching $M$.
Then $|P_v| \geq |P'_v|$ for any star $S'_v$ s.t. $P'_v \cap M = \emptyset$.
\label{lemma-star-greedy}
\end{lemma}

\begin{proof}
Assume for contradiction, $|P'_v| > |P_v|$ for some star $S'_v$ s.t. $P'_v \cap M = \emptyset$.
Since $|P'_v| > |P_v|$, $c(v) > |P_v|$.
For any trip $u \in N^{in}_{\mbox{-}M}(v)$, let $g_{i_u}(v)$ be the group s.t. $u \in g_{i_u}(v)$.
Let $u \in P'_v \setminus P_v$.
Note that $s_u \neq s_v$; otherwise, $u$ would have been included in $P_v$ by the greedy algorithm, and hence, $\delta_v > 0$.
From $c(v) > |P_v|$ and $\delta_v > 0$, the greedy algorithm must have executed the while-loop and checked all the groups in decreasing order of their size, and $\delta P_v = \delta_v$ at the end of the while-loop.
Because $c(v) > |P_v|$, $|P_v \cap g_{i_w}(v)| = |g_{i_w}(v)| \geq |P'_v \cap g_{i_w}(v)|$ for any $w \in P'_v \cap P_v$.
Since groups are checked in decreasing order of their size, $|P_v \cap g_{i}(v)| \geq |P'_v \cap g_{i_u}(v)|$ for every group $g_{i}(v)$ and every $u \in P'_v \setminus P_v$.
Recall that $\delta P_v = \delta_v$.
Hence, $|P_v| \geq |P'_v|$, which is a contradiction.
\end{proof}

\begin{definition}
A star $S_v$ rooted at $v$ is an improvement with respect to matching $M$ if $|P_v| \leq c(v), \delta P_v \leq \delta_v$ and $|S_v| - \sum_{(u,v) \in E(S_v)} |M(u)| > |M(v)|$.
\label{def-improvement}
\end{definition}

Definition~\ref{def-improvement} is equivalent to the original definition in~\cite{ESA17-K}, except the former is for the unweighted case and stop constraint. When an improvement is found, the current matching $M$ is increased by exactly $|S_v| - \sum_{(u,v) \in E(S_v)} |M(u)|$ edges.
For a vertex $v$ and a subset $S \subseteq E^{in}(v)$, let $N^{in}_{S}(v) = \{u \mid (u,v) \in S\}$.

\begin{lemma}
Let $M$ be the current matching and $v$ be a vertex with no improvement.
Let $S_v \subseteq E^{in}(v)$ s.t. $|S_v| \leq c(v)$ and $\delta P_v \leq \delta_v$, then $|S_v| \leq |M(v)| + |M(N^{in}_{S_v}(v))|$.
Further, if the star $S_v$ found by Algorithm~\arabic{algcounter} w.r.t. $M$ is not an improvement, then no other $S'_v$ is an improvement.
\label{lemma-unweighted}
\end{lemma}

\begin{proof}
When no improvement exists for a vertex $v$, we get 
$|S_v| - |M(N^{in}_{S_v}(v))| = |S_v| - \sum_{(u,v) \in S_v} |M(u)| \leq |M(v)|$ by Definition~\ref{def-improvement}.

To maximize $|S_v|$, we need to maximize $|S_v| - \sum_{(u,v) \in S_v} |M(u)|$, which can be done by selecting only in-neighbors of $v$ that are not in matching $M$.
This is because for any $(u,v) \in S_v$ s.t. $u$ is incident to a matched edge, $|M(u)| \geq 1$. In other words, including such a vertex $u$ cannot increase $|S_v| - \sum_{(u,v) \in S_v} |M(u)|$.
Algorithm~\arabic{algcounter} considers only in-neighbors $N^{in}_{\mbox{-}M}(v) = \{i \mid i \in N^{in} \setminus V(M)\}$.
By Lemma~\ref{lemma-star-greedy}, $|P_v|$ is maximized among all stars rooted at $v$ w.r.t. $M$.
Hence, lemma holds.
\end{proof}

Lemma~\ref{lemma-unweighted} is equivalent to Lemma 5 of~\cite{ESA17-K}, except the former is for the unweighted case and stop constraint.
By Lemma~\ref{lemma-unweighted} and the same argument of Lemma 6 in~\cite{ESA17-K}, we have the following lemma.

\begin{lemma}
The modified StarImprove algorithm computes a solution to an instance of ridesharing problem with stop constraint with
$\frac{1}{2}$-approximation.
\label{lemma-star-approximation}
\end{lemma}

\begin{theorem}
Let $(G, R)$ be a ridesharing instance satisfying all conditions, except condition (4). Let $|S^*|$ be the minimum number of drivers for $(G, R)$, $l = |R|$ and $K = \max_{i \in R} {n_i}$. Then,
\begin{itemize}
    \item The EdgeSwap algorithm computes a solution $(\sigma, S)$ for $(G, R)$ s.t. $|S^*| \leq |S| \leq \frac{K+2}{2} |S^*|$ with running time $O(M + l^{2K})$.
    \item The modified StarImprove algorithm computes a solution $(\sigma, S)$ for $(G, R)$ s.t. $|S^*| \leq |S| \leq \frac{K+2}{2} |S^*|$ with running time $O(M + K\cdot l^3)$, where $M$ is the size of a ridesharing instance which contains a road network and $l$ trips.
    \end{itemize}
\label{claim-app-ratio}
\end{theorem}

\begin{proof}
First, we need to construct a directed graph $H_R$ to represent the serve relation of the trips in $R$ as described in~\cite{GLZ19}, which takes $O(M)$ time.
Then reverse the direction of all arcs in $H_R$, and this gives an instance can be solved by the EdgeSwap and modified StarImprove algorithms.
Then, the first bullet point of the Lemma is due to the EdgeSwap paragraph stated above.
The rest of the proof is for the second bullet point.

By Lemma~\ref{lemma-star-approximation}, the modified StarImprove algorithm finds a solution to $(G, R)$ with at least $(l - |S^*|) / 2$ passengers, and hence, at most $|S| \leq l - (l - |S^*|) / 2 = (l + |S^*|) / 2$ drivers.
There are $l - |S^*|$ passengers in the optimal solution, implying $|S^*| \geq (l - |S^*|)/K = l / (K+1)$, so $l \leq (K+1)|S^*|$.
Therefore,
\begin{align*}
|S| \leq (l + |S^*|) / 2 \leq ((K+1)|S^*| + |S^*|) / 2 = (K+2)|S^*| / 2
\end{align*}
The original StarImprove algorithm has a for-loop to check each vertex $v$ to see if an improvement can be found, that is, it takes $O(l)$ time to check all in-neighbors of $v$ to see if a star $S_v$ that can increase $|M|$ exists, where $M$ is the current matching. In total, the for-loop takes $O(l^2)$ time.
Then for the modified StarImprove, it takes $O(K \cdot l^2)$ time; $O(K \cdot l)$ time for computing $S_v$ if it exists.
After an improvement is made each time, StarImprove scans every vertex again to check for another improvement until no improvement can be found, and this takes $O(l)$ time due to at most $O(l)$ improvements can be made for the unweighted case.
Thus, in total, the modified StarImprove has a running time of $O(M + K \cdot l^3)$.
\end{proof}

\section{A more practical new approximation algorithm} \label{sec-new-alg}
In this section, we present our new approximation algorithm for the ridesharing problem with stop constraint; our algorithm has a better running time than the approximation algorithms based on MCMP.
For our proposed algorithm, we assume the serve relation is transitive, that is, trip $i$ can serve trip $j$ and $j$ can serve trip $k$ imply $i$ can serve $k$.
In general, if each trip has a unique preferred path and trip $i$ can serve trip $j$ implies $j$'s preferred path is a subpath of $i$'s preferred path, then the serve relation is transitive.
For the scenarios we described in the paragraph Application of ridesharing in Section~\ref{sec-intro}, each driver always wants to use a unique preferred path $P$. In most cases, such a path $P$ is a shortest path.
For school commute, the length of $P$ is usually not too long. There is a high chance that the shortest path of any two points within $P$ is a subpath of $P$.
In practice, drivers may not even specify a preferred path, so a ridesharing instance satisfying the transitive serve relation may consist of trips with the preferred paths computed by a coordinator, or the road network has a unique shortest path between every pair of nodes and each trip uses the shortest path from the source to destination as the preferred path.

\subsection{New approximation algorithm}
Given a ridesharing instance $(G, R)$, we construct a directed meta graph $\Gamma(V, E)$ to express the serve relation, where $V(\Gamma)$ represents the start locations of all trips in $(G, R)$. Each node $\mu$ of $V(\Gamma)$ contains all trips with the same start location $\mu$.
There is an arc $(\mu, \nu)$ in $E(\Gamma)$ if a trip in $\mu$ can serve a trip in $\nu$. Since Conditions (1-3) and (5) are satisfied, if one trip in $\mu$ can serve a trip in $\nu$, any trip in $\mu$ can serve any trip in $\nu$. We say node $\mu$ can serve node $\nu$.
An arc $(\mu, \nu)$ in $\Gamma$ is called a \emph{short cut} if after removing $(\mu, \nu)$ from $\Gamma$, there is a path from $\mu$ to $\nu$ in $\Gamma$. We simplify $\Gamma$ by removing all short cuts from $\Gamma$.
The graph $\Gamma$ may contain a number of connected components. However, a trip in a node $\mu$ from one component cannot serve a trip in $\nu$ from another component and vice versa. Hence, the solution for a component is independent from another component in $\Gamma$.
In what follows, we assume $\Gamma$ is a single connected component and use $\Gamma$ for the simplified meta graph.
Notice that $\Gamma$ is an inverse tree and for every pair of nodes $\mu$ and $\nu$ in $\Gamma$, if there is a path from $\mu$ to $\nu$ then $\mu$ can serve $\nu$.
We label the nodes of $\Gamma$ as $V(\Gamma) = \{\mu_p, \mu_{p-1},...,\mu_1\}$, where $p = |V(\Gamma)|$, in such a way that for every arc $(\mu_b,\mu_a)$ of $\Gamma$, $b > a$, and we say $\mu_b$ has a larger label than $\mu_a$.
The labeling is done by the procedure in~\cite{GLZ19} (see Figure~\ref{fig-label}).
Figure~\ref{fig-metagraph} shows an example of a graph $\Gamma(V, E)$.
Each node in $\Gamma$ without an incoming arc is called an \emph{origin}, and $\mu_1$ is the unique sink.
For a node $\mu$ in $V(\Gamma)$, the set of trips contained in mode $\mu$ is denoted by $R(\mu)$.
For a set $U$ of nodes in $V(\Gamma)$, $R(U) = \bigcup_{\mu \in U} R(\mu)$.
Similarly, given a set $S$ of drivers, we denote the set of drivers in the nodes
of $U$ by $S(U)$ and the set of drivers in a node $\mu$ by $S(\mu)$.
For a trip $i \in R$, the node that contains $i$ is denoted by $\node(i)$, that is, if $i \in R(\mu)$ then $\node(i) = \mu$.
Table~\ref{table-def-5.2} contains the basic notation and definition for this section.

\begin{figure}[htbp]
\small
\begin{center}
\begin{tabbing}
XX\=XX\=XX\=XX\=XX\=XX\=XX\=XX\=XX\=XX\=XX\=XX\=    \kill
{\bf Procedure} Label-Inverse-Tree\\
{\bf Input:} An inverse tree $\Gamma$ of $l$ trips and $p$ nodes.\\
{\bf Output:} Distinct integer labels $\mu_p,\ldots,\mu_1$ for nodes in $\Gamma$.\\
{\bf begin}\\
\>let ST be a stack and push the sink of $\Gamma$ into ST;\\
\>$i:=p$ and mark every arc in $\Gamma$ un-visited;\\
\>{\bf while} ST$\neq \emptyset$ {\bf do} \\
\>\>let $\mu$ be the node at the top of ST;\\
\>\>{\bf if} there is an arc $(\nu,\mu)$ in $\Gamma$ un-visited {\bf then}\\
\>\>\>push $\nu$ into ST and mark $(\nu,\mu)$ visited;\\
\>\>{\bf else}\\
\>\>\>remove $\mu$ from ST; assign $\mu$ integer label $\mu_i$; $i:=i-1$;\\
\>\>{\bf endif}\\
\>{\bf endwhile}\\
{\bf end.}
\end{tabbing}
\caption{Procedure for assigning integer labels to nodes in $\Gamma$~\cite{GLZ19}.}
\label{fig-label}
\end{center}
\end{figure}

\begin{figure}[htbp!]
\centering
\includegraphics[scale=0.63]{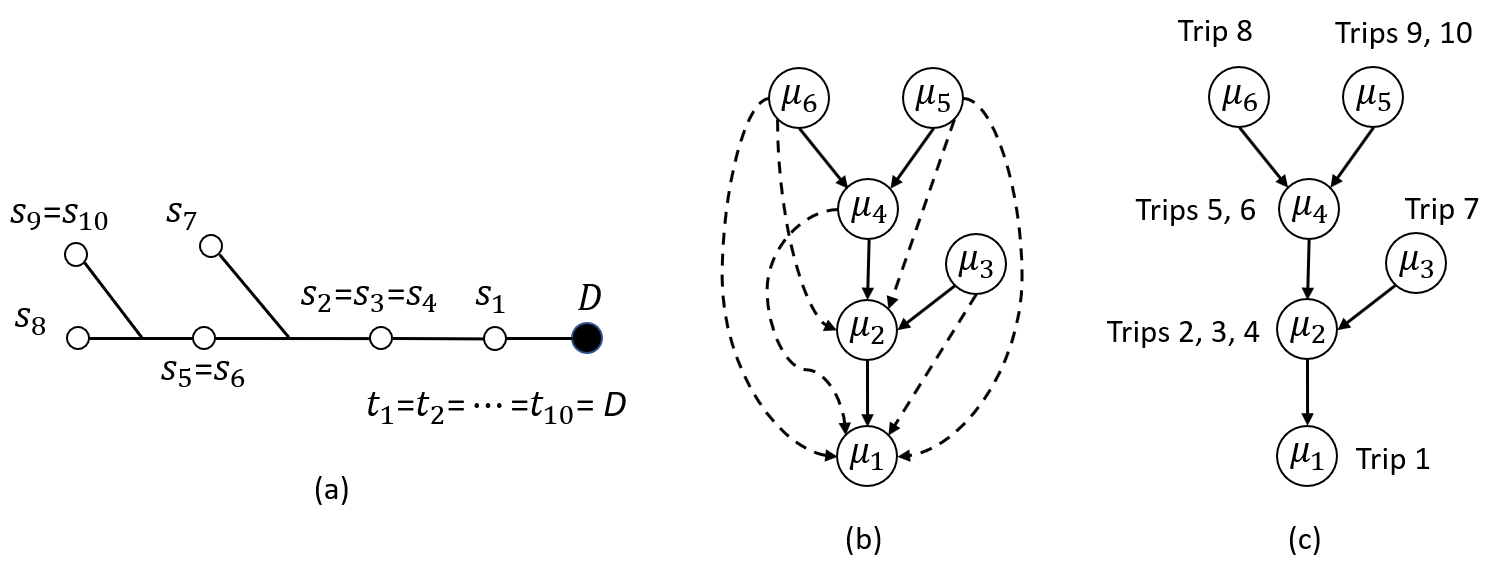}
\caption{(a) A set $R$ of 10 trips with same destination $D$ in the road network graph $G$. (b) The directed meta graph expressing the serve relation of these trips with shortcuts in dashed arcs. (c) The simplified meta graph, which is an inverse tree.}
\label{fig-metagraph}
\vspace{1mm}
\end{figure}

\begin{table}[ht]
\begin{center}
   \begin{tabular}{| l | l |}
   	\hline
   	\textbf{Notation} & \textbf{Definition}                                                    \\ \hline
    $\Gamma(V, E)$ and $p$ & A directed graph expressing the serve relation and $p = |V(\Gamma)|$ \\
    $\mu$ is an ancestor of $\nu$ & If $\exists$ a nonempty path from $\mu$ to $\nu$ in $\Gamma$ ($\nu$ is a descendant of $\mu$) \\ 
  	$A_{\mu}$ and $A^*_{\mu}$ & Set of ancestors of $\mu$ and $A^*_{\mu} = A_{\mu} \cup \{\mu\}$ respectively \\
    $D_{\mu}$ and $D^*_{\mu}$ & Set of descendants of $\mu$ and $D^*_{\mu} = D_{\mu} \cup \{\mu\}$ respectively \\
    $R(\mu)$ and $R(U)$ & Set of trips in a node $\mu$ and in a set $U$ of nodes respectively \\
    $S(\mu)$ and $S(U)$ & Set of drivers in a node $\mu$ and in nodes $U$ respectively \\
    $\node(i)$ & The node that contains trip $i$ (if $i \in R(\mu)$ then $\node(i) = \mu$) \\
    $\free(i)$ & The remaining seats (capacity) of $i$ w.r.t. current solution $(S, \sigma)$ \\
    $\stopp(i)$ & The number of stops $i$ has to made to serve all trips assigned to $i$ \\ \hline
   \end{tabular}
\caption{Basic notation and definition used in this section.}
\label{table-def-5.2}
\end{center}
\vspace{-1mm}
\end{table}

\stepcounter{algcounter}
We divide all trips of $R$ into two sets $W$ and $X$ as follows:
\begin{align*}
&W = \{i \in R \mid n_i = 0\} \cup \{i \in R(\mu) \mid \delta_i = 0 \text{ and } |R(\mu) = 1| \text{ for every node } \mu \in V(\Gamma)\} \text{ and} \\
&X = R \setminus W.
\end{align*}
For a node $\mu$ in $\Gamma$, let $X(\mu) = X \cap R(\mu)$ and $W(\mu) = W \cap R(\mu)$.
Our algorithm tries to minimize the number of drivers that only serve itself.
There are three phases in the algorithm. In Phase-I, it serves all trips of $W$ and tries to minimize the number of trips in $W$ that are assigned as drivers since each trip of $W$ can serve only itself. Let $Z$ be the set of unserved trips after Phase-I such that for every $i \in Z$, $\delta_i=0$. In Phase-II, it serves all trips of $Z$ and tries to minimize the number of trips in $Z$ to be assigned as drivers, each only serves itself. In Phase III, it serves all remaining trips.
Let $(S, \sigma)$ be the current partial solution and $i \in R$ be a driver. Denoted by $\free(i) = n_i - |\sigma(i)| + 1$ is the remaining seats (capacity) of $i$ with respect to solution $(S, \sigma)$.
Denoted by $\stopp(i)$ is the number of stops $i$ has to made in order to serve all trips in $\sigma(i)$ w.r.t. $(S, \sigma)$.
For the initial solution $(S, \sigma) = (\emptyset, \emptyset)$, $\free(i) = n_i$ and $\stopp(i) = 0$ for all $i \in R$.
For a driver $i$ and node $\mu$, we define $R(i, \mu, S)$ as the set of $\min\{\free(i), |R(\mu) \setminus \sigma(S)|\}$ trips in $R(\mu) \setminus \sigma(S)$ and $W(i, \mu, S)$ as the set of $\min\{\free(i), |W(\mu) \setminus \sigma(S)|\}$ trips in $W(\mu) \setminus \sigma(S)$, and similarly for $Z(i, \mu, S)$.
The three phases of the approximation algorithm (Algorithm~\arabic{algcounter}) are described in following, and the pseudo code is given in Figure~\ref{fig-alg-1}.
\vspace{2mm}

\noindent \textbf{(Phase-I)} In this phase, the algorithm assigns a set of drivers to serve all trips of $W$, and it ends once all trips of $W$ are served.
Let $\Gamma(W) = \{\mu \in V(\Gamma) \mid W(\mu) \setminus \sigma(S) \neq \emptyset\}$, and in each iteration, a node of $\Gamma(W)$ is processed.
In each iteration, the node $\mu = \text{argmax}_{\mu \in \Gamma(W)} |W(\mu) \setminus \sigma(S)|$ is selected and a subset of trips in $W(\mu) \setminus \sigma(S)$ is served by a driver as follows:
\begin{itemize}
\item Let $\hat{X}_1 = \{i \in S(A_{\mu}) \mid \free(i) > 0 \wedge \stopp(i) < \delta_i \}$ and 
$\bar{X} = \{i \in X \cap R(A^*_{\mu}) \setminus \sigma(S) \mid \stopp(i) < \delta_i \vee i \in R(\mu)\}$.
The algorithm finds and assigns a trip $x$ as a driver to serve $W(x, \mu, S)$ such that $x = \text{argmin}_{x \in \hat{X}_1 \cup \bar{X} \text{ : } n_x \geq |W(\mu) \setminus \sigma(S)|} \delta_x - \stopp(x)$.
    \begin{itemize}
    \item If such a trip $x$ does not exist, it means that $n_x < |W(\mu) \setminus \sigma(S)|$ for every $x \in \hat{X}_1 \cup \bar{X}$ assuming $\hat{X}_1 \cup \bar{X} \neq \emptyset$. Then, $x = \text{argmax}_{x \in \hat{X}_1 \cup \bar{X}} \free(x)$ is assigned as a driver to serve $W(x, \mu, S)$.
    If there is more than one $x$ with same $\free(x)$, the trip with smallest $\delta_x - \stopp(x)$ is selected.
    \end{itemize}
\item When $\hat{X}_1 \cup \bar{X} = \emptyset$, assign every $w \in W(\mu) \setminus \sigma(S)$ as a driver to serve itself.
\end{itemize}
 
\noindent \textbf{(Phase-II)} In the second phase, all trips of $Z = \{i \in R \setminus \sigma(S) \mid \delta_i = 0\}$ are served.
Let $\Gamma(Z) = \{\mu \in \Gamma \mid Z(\mu) = (Z \cap R(\mu)) \neq \emptyset\}$.
Each node $\mu$ of $\Gamma(Z)$ is processed in the decreasing order of their node labels.
\begin{itemize}
    \item If $|Z(\mu)| \geq 2$, trip $x = \text{argmax}_{x \in Z(\mu)} n_x$ is assigned as a driver and serves $Z(x, \mu, S)$ consists of trips with smallest capacity among trips in $Z(\mu) \setminus \sigma(S)$.
    \item This repeats until $|Z(\mu)| \leq 1$. Then next node in $\Gamma(Z)$ is processed.
\end{itemize}
After all nodes of $\Gamma(Z)$ are processed, each non-empty node $\mu$ of $\Gamma(Z)$ is processed again; note that every $\mu$ contains exactly one $z \in Z(\mu)$ now, that is, $|Z(\mu)|=1$.
\begin{itemize}
    \item A driver $x \in \hat{X}_2 = \{i \in S(A^*_{\mu}) \mid \free(i) > 0 \wedge (\stopp(i) < \delta_i \vee i \in R(\mu))\}$ with largest $\free(x)$ is selected to serve $z = Z(\mu)$ if $\hat{X}_2 \neq \emptyset$.
    \item If $\hat{X}_2 = \emptyset$, a trip $x \in \bar{X} = \{i \in X \cap R(A^*_{\mu}) \setminus \sigma(S) \mid \stopp(i) < \delta_i \vee i \in R(\mu)\}$ with largest $\delta_x$ is selected to serve $z = Z(\mu)$.
\end{itemize}

\noindent \textbf{(Phase-III)} To serve all remaining trips, the algorithm processes each node of $\Gamma$ in decreasing order of node labels from $\mu_p$ to $\mu_1$.
Let $\mu_j$ be the node being processed by the algorithm.
Suppose there are trips in $R(\mu_j)$ that have not be served, that is, $R(\mu_j) \nsubseteq \sigma(S)$.
\begin{itemize}
    \item A driver $x \in \hat{X}_2 = \{i \in S(A^*_{\mu_j}) \mid \free(i) > 0 \wedge (\stopp(i) < \delta_i \vee i \in R(\mu_j))\}$ with largest $\free(x)$ is selected if $\hat{X}_2 \neq \emptyset$.
    \item If $\hat{X}_2 = \emptyset$, a trip $x = \text{argmax}_{x \in X(\mu_j) \setminus \sigma(S)} n_x$ is assigned as a driver. If the largest $n_x$ is not unique, the trip with the smallest $\delta_x$ is selected.
    \item In either case, $x$ is assigned to serve $R(x, \mu_j, S)$. This repeats until all of $R(\mu_j)$ are served.
    Then, next node $\mu_{j-1}$ is processed.
\end{itemize}

\begin{figure}[!htbp]
\small
\textbf{Algorithm~\arabic{algcounter}} New approximation algorithm \\
\textbf{Input:} A ridesharing instance $(G,R)$ and the meta graph $\Gamma$ (inverse tree) for $(G,R)$. \\
\textbf{Output:} A solution $(S,\sigma)$ for $(G,R)$ with $\frac{K+2}{2}$-approximation ratio.
\begin{algorithmic}[1]
\\ $(S, \sigma) = (\emptyset, \emptyset)$. Let $\Gamma(W) = \{\mu \in V(\Gamma) \mid W(\mu) \setminus \sigma(S) \neq \emptyset\}$.
\While {$\Gamma(W) \neq \emptyset$}			\hspace{2.6cm} /* Beginning of Phase-I */
    \State Compute $\mu = \text{argmax}_{\mu \in \Gamma(W)} |W(\mu) \setminus \sigma(S)|$.
        Let $\hat{X}_1 = \{i \in S(A_{\mu}) \mid \free(i) > 0 \wedge \stopp(i) < \delta_i \}$
    \State and $\bar{X} = \{i \in X \cap R(A^*_{\mu}) \setminus \sigma(S) \mid \stopp(i) < \delta_i \vee i \in R(\mu)\}$.
        \If {$\hat{X}_1 \cup \bar{X} \neq \emptyset$}
            \State Compute $x = \text{argmin}_{x \in \hat{X}_1 \cup \bar{X} \text{ : } n_x \geq |W(\mu) \setminus \sigma(S)|} \delta_x - \stopp(x)$.
            \State \textbf{if} $x == \emptyset$ \textbf{then} $x = \text{argmax}_{x \in \hat{X}_1 \cup \bar{X}} \free(x)$ with the smallest $\delta_x - \stopp(x)$.
            \State \textbf{if} $x \notin S$ \textbf{then} $S = S \cup \{x\}$; $\sigma(x) = \{x\}$;
            \State $\sigma(x) = \sigma(x) \cup W(x, \mu, S)$; update $\free(x)$ and $\stopp(x)$;
		\Else
            \State for each $w \in W(\mu) \setminus \sigma(S)$, $S = S \cup \{w\}$, $\sigma(w) = \{w\}$; update $\free(w)$;
		\EndIf
\EndWhile \hspace{2.6cm} /* End of Phase-I. Below is Phase-II */

\\ Let $Z = \{i \in R \setminus \sigma(S) \mid \delta_i = 0\}$ and $\Gamma(Z)$ be the set of nodes containing $Z$.
\For {each node $\mu \in \Gamma(Z)$ in decreasing order of the node labels}
    \While {$|Z(\mu)| \geq 2$}
        \State Compute $x = \text{argmax}_{x \in Z(\mu)} n_x$. $S = S \cup \{x\}$; $\sigma(x) = \{x\}$;
        $\sigma(x) = \sigma(x) \cup Z(x, \mu, S)$ where
        \State $Z(x, \mu, S)$ consists of trips with smallest capacity; update $\free(x)$ and $\stopp(x)$; update $Z$.
    \EndWhile
\EndFor
\For {each node $\mu \in \Gamma(Z)$ in decreasing order of node labels} \hspace{0.2cm} /* implying $|Z(\mu)| = 1$ */
    \State Let $\hat{X}_2 = \{i \in S(A^*_{\mu}) \mid \free(i) > 0 \wedge (\stopp(i) < \delta_i \vee i \in R(\mu))\}$.
    \State \textbf{if} {$\hat{X}_2 \neq \emptyset$} \textbf{then} 
        Compute $x = \text{argmax}_{x \in \hat{X}_2} \free(x)$.
    \State \textbf{else} Let $\bar{X} = \{i \in X \cap R(A^*_{\mu}) \setminus \sigma(S) \mid \stopp(i) < \delta_i \vee i \in R(\mu)\}$.
        Compute $x = \text{argmax}_{x \in \bar{X}} \delta_x$.
    \State \textbf{if} $x \notin S$ \textbf{then} $S = S \cup \{x\}$; $\sigma(x) = \{x\}$;
    \State $\sigma(x) = \sigma(x) \cup Z(x, \mu, S)$; update $\free(x)$ and $\stopp(x)$;
\EndFor \hspace{2.8cm} /* End of Phase-II. Below is Phase-III */
\For {each node $\mu$ from $\mu_p$ to $\mu_1$}		
    \While {$R(\mu) \nsubseteq \sigma(S)$}
        \State Let $\hat{X}_2 = \{i \in S(A^*_{\mu}) \mid \free(i) > 0 \wedge (\stopp(i) < \delta_i \vee i \in R(\mu))\}$.
        \State \textbf{if} {$\hat{X}_2 \neq \emptyset$} \textbf{then}
            Compute $x = \text{argmax}_{x \in \hat{X}_2} \free(x)$.
        \State \textbf{else} 
            Compute $x = \text{argmax}_{x \in X(\mu) \setminus \sigma(S)} n_x$ (with smallest $\delta_x$ as a tiebreaker)
        \State \textbf{if} $x \notin S$ \textbf{then} $S = S \cup \{x\}$; $\sigma(x) = \{x\}$;
        \State $\sigma(x) = \sigma(x) \cup R(x, \mu, S)$; update $\free(x)$ and $\stopp(x)$;
    \EndWhile
\EndFor

\end{algorithmic}
\caption{Algorithm for approximating the minimum number of drivers.}
\label{fig-alg-1}
\end{figure}

\subsection{Analysis of new approximation algorithm}
A driver in a solution is called a \textit{solo} driver if it serves only itself. Algorithm 2 tries to minimize the number of solo drivers. 
Recall that $W$ is the set of trips, each of which can serve only itself. The algorithm, in Phase-I, computes a partial solution to serve all trips of $W$ and tries to assign as few trips of $W$ to be drivers as possible.
In Phase-II, the set $Z$ of unserved trips after Phase-I (every $i \in Z$ has $\delta_i = 0$) is served. The rationale to serve such set of trips is that many trips of $Z$ can become solo drivers if all trips of $R(\node(i)) \setminus \{i\}$ for $i \in Z$ are served before $i$ is processed or considered.
This can cause $Z$ to have the same characteristic as $W$, so we need to treat $Z$ separately.
Let $\lambda$ be the number of solo drivers in a solution computed by Algorithm 2 and $\lambda^*$ be the number of solo drivers in any optimal solution. Then there are at most $(|R|-\lambda)/2+\lambda$ drivers in the solution computed by Algorithm 2 and at least $(|R|-\lambda^*)/(K+1)+\lambda^*$ drivers in the optimal solution. A central line of the analysis is to show that $\lambda^*$ is close to $\lambda$ which guarantees the approximation ratio of Algorithm 2

We now introduce some notation used in our analysis.
Denoted by $(S, \sigma)$ is the complete solution computed by Algorithm~\arabic{algcounter}.
Denoted by $(S_{\RNum{1}}, \sigma_{\RNum{1}})$ is the partial solution computed at the end of Phase-I, so all trips of $W$ are served by drivers in $S_{\RNum{1}}$.
For every driver $i \in S_{\RNum{1}}$, $(\sigma_{\RNum{1}}(i) \setminus \{i\}) \cap (R \setminus W) = \emptyset$.
Let $S_{\RNum{1}}(X) = S_{\RNum{1}} \cap X$ and $S_{\RNum{1}}(W) = S_{\RNum{1}} \cap W = S_{\RNum{1}} \setminus S_{\RNum{1}}(X)$.
Note that each driver $i \in S_{\RNum{1}}(X)$ must serve at least one trip from $W$ and $\sigma_{\RNum{1}}(S_{\RNum{1}}(X)) \setminus S_{\RNum{1}}(X) \subseteq W$ if $S_{\RNum{1}}(X) \neq \emptyset$.
Let $W = \{W_1, \ldots, W_e\}$ such that each $W_j$ ($1 \leq j \leq e$) is the set of trips served by a driver (or drivers when $W_j \subseteq W$) in $S_{\RNum{1}}$ for iteration $j$, where $e$ is the last iteration of Phase-I.
For each $W_j$, $W_j$ is a subset of $W(\mu_{a_j})$ for some node $\mu_{a_j}$ (indexed at $a_j$), and let $(S_j, \sigma_j)$ be the partial solution just after serving $W_j$, $1 \leq j \leq e$.
For a driver $i \in S_j$, $\free_j(i) = n_i - |\sigma_j(i)| + 1$ is the remaining available seats (capacity) of $i$ w.r.t. $(S_j, \sigma_j)$,
and $\stopp_j(i)$ is the number of stops $i$ has to made in order to serve all trips in $\sigma_j(i)$ w.r.t. $(S_j, \sigma_j)$.

\begin{property}
For every trip $i$ that is assigned as a driver, $i$ remains a driver until the algorithm terminates and $\free(i)$ is non-increasing throughout the algorithm.
\label{property-1}
\end{property}

Recall that each set $W_j$ of trips either are served by one driver or $W_j \subseteq S_{\RNum{1}}(W)$.
For clarity, we denote each set $W_j \subseteq S_{\RNum{1}}(W)$ by $\tilde{W}_j$.
When trips of $\tilde{W}_j$ are assigned as drivers to serve themselves, all other trips $W(\mu_{a_j}) \setminus \tilde{W}_j$ must have been served by drivers in $S_{j-1}$ such that no driver in $S_{j-1}$ or trip in $X \setminus \sigma_{j-1}(S_{j-1})$ can serve $\tilde{W}_j$.
In other words, $\tilde{W}_j = W(\mu_{a_j}) \setminus \sigma_{j-1}(S_{j-1})$ and $\bar{X}_1 \cup \hat{X} = \emptyset$ w.r.t. $(S_{j-1}, \sigma_{j-1})$, so the algorithm has the following property.
\begin{property}
For every pair $\tilde{W}_i$ and $\tilde{W}_j$ ($i\neq j$), $\mu_{a_i} \neq \mu_{a_j}$.
\label{property-2}
\end{property}

Suppose $(S^*, \sigma^*)$ is an optimal solution for $(G,R)$ with $|S^*|$ minimized.
We first show, in Lemma~\ref{lemma-phase-one-W}, that the number of trips in $S_{\RNum{1}}(W)$ served by $S^*$ is at most that of the passengers served by $\sigma_{\RNum{1}}(S_{\RNum{1}}(X))$.
The proof idea is as follows.
Let $U \subseteq S^*$ be the set of drivers such that for every $u \in U$, $\sigma^*(u) \cap S_{\RNum{1}}(W) \neq \emptyset$ and $U \cap W = \emptyset$.
We prove that $U$ are also drivers in $S_{\RNum{1}}$ (specifically, $U \subseteq S_{\RNum{1}}(X)$) and $\sigma_{\RNum{1}}(u)$ serves at least $|\sigma^*(u) \cap S_{\RNum{1}}(W)|$ passengers for each $u \in U$.

\begin{lemma}
Let $(S^*, \sigma^*)$ be an optimal solution for $(G,R)$ and $S^*(W) = W \cap S^*$.
Let $U \subseteq S^*$ be the set of drivers that serve all trips of $W \setminus S^*(W)$.
Then $|\sigma_{\RNum{1}}(S_{\RNum{1}}(X)) \setminus S_{\RNum{1}}(X)| \geq |\sigma^*(U) \cap S_{\RNum{1}}(W)|$.
\label{lemma-phase-one-W}
\end{lemma}

\begin{proof}
Let $U_j$ be the set of drivers in $S^*$ that serve $(W_1 \cup \ldots \cup W_j) \setminus S^*(W)$ for $1 \leq j \leq e$.
Note that $W = W_1 \cup \ldots \cup W_e$ and $U_e = U$.
Let $\tilde{W}_{a_1},\ldots, \tilde{W}_{a_d}$ be the sets computed by the algorithm such that $\tilde{W}_{a_b} \subseteq S_{\RNum{1}}(W), 1 \leq b \leq d$, and for $1 \leq b < c\leq d$, $\tilde{W}_{a_b}$ is computed before $\tilde{W}_{a_c}$. For each $\tilde{W}_{a_b}$, the drivers of $U_{a_b}$ that serve $\tilde{W}_{a_b} \setminus S^*(W)$ can be partitioned into two sets: (1) $U'_{a_b} = \{u \in U_{a_b} \mid \sigma^*(u) \cap \tilde{W}_{a_b} \neq \emptyset \text{ and } u \in R(\mu_{a_b})\}$ and
(2) $U''_{a_b} = \{u \in U_{a_b} \mid \sigma^*(u) \cap \tilde{W}_{a_b} \neq \emptyset \text{ and } u \in R(A_{\mu_{a_b}})\}$.
We consider them separately.

(1) Due to $W(\mu_{a_b}) \neq \emptyset$ ($\mu_{a_b} \in \Gamma(W)$), 
the algorithm must have already assigned every $u \in U'_{a_b}$ as a driver in $S_{a_{b-1}}(X)$ when $\mu_{a_b}$ is processed since such a trip $u$ must be included in $\bar{X}$ w.r.t. the partial solution just before $\mu_{a_b}$ is processed.
Further, it must be that $\free_{a_{b-1}}(u) = 0$. Otherwise, $\sigma_{a_{b-1}}(u)$ would have served trips from $\tilde{W}_{a_b}$, a contradiction to the algorithm.
From $\free_{a_{b-1}}(u) = 0$, $|\sigma_{a_b}(u) \cap W| \geq |\sigma^*(u) \cap W|$ for every $u \in U'_{a_b}$, that is, $|\bigcup_{u \in U'_{a_b}} \sigma_{a_b}(u) \cap W| \geq |\bigcup_{u \in U'_{a_b}} \sigma^*(u) \cap W|$.

(2) Every $u \in U''_{a_b}$ must also be a driver in $S_{a_{b-1}}(X)$ with $\free_{a_b}(u) < n_u$. Otherwise, $u$ would have been assigned (from unassigned) as a driver in $S_{a_b}$ to serve trips from $\tilde{W}_{a_b}$.
We further divide $U''_{a_b} $into two subsets: $U''_{a_b}(0) = \{u \in U''_{a_b} \mid \free_{a_b}(u) = 0\}$ and $U''_{a_b}(1) = \{u \in U''_{a_b} \mid \free_{a_b}(u) \geq 1\}$.
We consider $U''_{a_b}(0)$ in case (2.1) and $U''_{a_b}(1)$ in case (2.2).

(2.1) For every $u \in U''_{a_b}(0)$, $|\sigma_{a_b}(u) \cap W| \geq |\sigma^*(u) \cap W|$ since $\free_{a_b}(u) = 0$. This implies that $|\bigcup_{u \in U''_{a_b}(0)} \sigma_{a_b}(u) \cap W| \geq |\bigcup_{u \in U''_{a_b}(0)} \sigma^*(u) \cap W|$.
(2.2) Consider any driver $u \in U''_{a_b}(1)$.
Let $W_j$ be a non-empty set of passengers served by $\sigma_{a_b}(u)$ where $j < a_b$. In other words, $W_j$ is computed before $\tilde{W}_{a_b}$.
Recall that $W_j \subseteq W(\mu_{a_j})$, $W_j$ are the only passengers in $W(\mu_{a_j})$ served by $u$, and $(S_{j-1}, \sigma_{j-1})$ is the partial solution just before trips of $W_j$ are served.
From $\free_{j-1}(u) > \free_{a_b}(u) > 0$, $W_j = W(\mu_{a_j}) \setminus \sigma_{j-1}(S_{j-1})$ must be served by $\sigma_j(u)$, implying $|W_j| < \free_{j-1}(u)$.
Since $W_j$ is computed before $\tilde{W}_{a_b}$, $|\tilde{W}_{a_b}| \leq |W(\mu_{a_b}) \setminus \sigma_{j-1}(S_{j-1})| \leq |W(\mu_{a_j}) \setminus \sigma_{j-1}(S_{j-1})| < \free_{j-1}(u)$, meaning $|W_j| \geq |\tilde{W}_{a_b}|$ for every set $W_j$ of passengers served by $\sigma_{a_b}(u)$.
From the proofs of Cases (1) and (2), we have the following property.
\begin{property}
Every $u \in U'_{a_b} \cup U''_{a_b}$ is also a driver in $S_I(X)$, that is, $U \subseteq S_I(X)$.
\label{prop-opt-drivers-in-p1}
\end{property}

Consider any pair $u_b \in U''_{a_b}(1)$ and $u_c \in U''_{a_c}(1)$ with $u_b \neq u_c$ for any $1 \leq b < c \leq d$.
Since $u_b \neq u_c$, the analysis of Case (2.2) can be applied to $u_b$ and $u_c$ independently, that is, $|W_{j_b}|\geq |\tilde{W}_{a_b}|$ for every set $W_{j_b}$ of passengers served by $\sigma_{j_b}(u_b)$, and $|W_{j_c}|\geq |\tilde{W}_{a_c}|$ for every set $W_{j_c}$ served by $\sigma_{j_c}(u_c)$.
Now, consider the case $u_b = u_c$. Assume that $U''_{a_b}(1)\cap U''_{a_c}(1) \neq \emptyset$ for some $1 \leq b < c \leq d$.
Consider any driver $u \in U''_{a_b}(1) \cap U''_{a_c}(1)$. By definition, $u$ serves trips from both $\tilde{W}_{a_b}$ and $\tilde{W}_{a_c}$.
Since $\free_{a_c}(u) > 0$, $\stopp_{a_c}(u) = \delta_u$.
It must be that $\free_{a_b}(u) \geq \free_{a_c}(u) > 0$ and $\stopp_{a_b}(u) = \delta_u$ (otherwise, $\sigma_{a_b}(u)$ would have served trips from $\tilde{W}_{a_b}$).
From this and $\mu_{a_b} \neq \mu_{a_c}$ (by Property~\ref{property-2}), $\delta_u \geq 2$ and $\sigma_{a_c}(u)$ serves
at least two sets $W_{j_b}$ and $W_{j_c}$ of passengers before $\tilde{W}_{a_b}$ is computed.
By the conclusion of previous paragraph (Case 2.2), $|W_{j_b}| \geq |\tilde{W}_{a_b}|$ and $|W_{j_c}|\geq |\tilde{W}_{a_c}|$.
This can be generalized to all sets $\tilde{W}_{a_1}, \ldots, \tilde{W}_{a_d} \subseteq S_{\RNum{1}}(W)$ such that trips of $\tilde{W}_{a_b} \setminus S^*(W)$ are served by $U_{a_b}$ for $1 \leq b \leq d$.
We get $|\bigcup_{u \in U'_{a_b} \cup U''_{a_b}, 1 \leq b \leq d} \sigma_{a_b}(u) \cap W| \geq |\bigcup_{u \in U'_{a_b} \cup U''_{a_b}, 1 \leq b \leq d} \sigma^*(u) \cap \tilde{W}_{a_b}|$.
By definition, $\bigcup_{u \in U'_{a_b} \cup U''_{a_b}, 1 \leq b \leq d} \sigma_{a_b}(u) \cap W = \sigma_{\RNum{1}}(U) \setminus U$ and $\bigcup_{u \in U'_{a_b} \cup U''_{a_b}, 1 \leq b \leq d} \sigma^*(u) \cap \tilde{W}_{a_b} = \sigma^*(U) \cap S_{\RNum{1}}(W)$.
Since $U \subseteq S_{\RNum{1}}(X)$ (Property~\ref{prop-opt-drivers-in-p1}), $|\sigma_{\RNum{1}}(S_{\RNum{1}}(X)) \setminus S_{\RNum{1}}(X)| \geq |\sigma_{\RNum{1}}(U) \setminus U| \geq |\sigma^*(U)\cap S_{\RNum{1}}(W)|$.
\end{proof}

\begin{lemma}
Let $(S^*, \sigma^*)$ be any optimal solution for $(G,R)$.
Let $F_{\RNum{1}}$ be the set of drivers in $S^*$ that serve all trips of $\sigma_{\RNum{1}}(S_{\RNum{1}})$ in $(S^*, \sigma^*)$ .
Then, $|F_{\RNum{1}}| \geq \frac{2|S_{\RNum{1}} \cup F_{\RNum{1}}|}{K+2}$.
\label{lemma-phase-one-min}
\end{lemma}

\begin{proof}
Three sets $U, B_1, B_2$ of drivers in $S^*$ are considered, each of which serves a portion of trips of $\sigma_{\RNum{1}}(S_{\RNum{1}})$ in $(S^*, \sigma^*)$, and altogether $\sigma^*(U \cup B_1 \cup B_2) \cup S^*(W) \supseteq \sigma_{\RNum{1}}(S_{\RNum{1}})$, where $S^*(W) = S^* \cap W$.
Let $U$ be the set of drivers in $S^*$ that serve all trips of $S_{\RNum{1}}(W) \setminus S^*(W)$ in $(S^*, \sigma^*)$.
By Property~\ref{prop-opt-drivers-in-p1} and Lemma~\ref{lemma-phase-one-W}, all of $U$ must be drivers in $S_{\RNum{1}}(X)$ and $|\sigma_{\RNum{1}}(U) \setminus U| \geq |\sigma^*(U) \cap S_{\RNum{1}}(W)|$.
In this proof, the drivers in $S_{\RNum{1}}$ are partitioned into three sets: $S_{\RNum{1}}(W), U$, and $S_X = S_{\RNum{1}} \setminus (S_{\RNum{1}}(W) \cup U)$.

It requires another set $B_1$ of drivers in $S^*$ to serve all trips of $(\sigma_{\RNum{1}}(U) \setminus U) \subseteq W$ in $(S^*, \sigma^*)$ because $\sigma_{\RNum{1}}(U) \cap S_{\RNum{1}}(W) = \emptyset$ and $|\sigma_{\RNum{1}}(U) \setminus U| \geq |\sigma^*(U) \cap S_{\RNum{1}}(W)|$.
From $|\sigma_{\RNum{1}}(U) \setminus U| \geq |\sigma^*(U) \cap S_{\RNum{1}}(W)|$, $\sigma_{\RNum{1}}(U) \cap S_{\RNum{1}}(W) = \emptyset$ and that $\sigma^*(U) \cap S_{\RNum{1}}(W) = S_{\RNum{1}}(W) \setminus S^*(W)$,
we have $|(S_{\RNum{1}}(W) \setminus S^*(W)) \cup (\sigma_{\RNum{1}}(U) \setminus U)| \geq 2|S_{\RNum{1}}(W) \setminus S^*(W)|$. Therefore, $|U \cup B_1|\geq 2|S_{\RNum{1}}(W) \setminus S^*(W)|/K$ is the minimum number of drivers required in $S^*$ to serve all of $(S_{\RNum{1}}(W) \setminus S^*(W)) \cup (\sigma_{\RNum{1}}(U) \setminus U)$.
In the worst case, the algorithm can assign each trip $v \in B_1$ to be a driver in $S \setminus S_{\RNum{1}}$ such that $v$ serves only itself.

Consider the remaining set of drivers $S_X = S_{\RNum{1}} \setminus (S_{\RNum{1}}(W) \cup U)$.
For each driver $x \in S_X$, $\sigma_{\RNum{1}}(x)$ must serve at least one trip from $W$, meaning $|\sigma_{\RNum{1}}(x)| \geq 2$ and $|\sigma_{\RNum{1}}(S_X)| \geq 2|S_X|$.
Let $B_2$ be the set of drivers in $S^*$ that serve all trips of $\sigma_{\RNum{1}}(S_X)$ in $(S^*, \sigma^*)$.
We now consider the size of $B_2$.
Note that $B_2 \cap S_X$ may or may not be empty.
In the worst case, each trip $v \in B_2 \setminus S_X$ can be assigned as a driver in $S \setminus S_{\RNum{1}}$ s.t. $v$ serves itself only.
Hence, the ratio between the number of drivers in $S$ that serve $\sigma_{\RNum{1}}(S_X) \cup B_2$ and $B_2$ is
$(|S_X| + |B_2 \setminus S_X|)/|B_2|$.
This function is monotone increasing in $|B_2 \setminus S_X|$.
Thus, $B_2 \cap S_X = \emptyset$ gives the worst case.
From this and $|\sigma^*(v) \cap \sigma_{\RNum{1}}(S_X)| \leq K$ for each driver in $v \in B_2$, $|B_2| \geq 2|S_X| / K$.
Since $\sigma_{\RNum{1}}(S_X) \cap \sigma_{\RNum{1}}(U) = \emptyset$ and $\sigma_{\RNum{1}}(S_X) \cap \sigma_{\RNum{1}}(W) = \emptyset$,
$|(S_{\RNum{1}}(W) \setminus S^*(W)) \cup (\sigma_{\RNum{1}}(U) \setminus U) \cup \sigma_{\RNum{1}}(S_X)| \geq 2|S_{\RNum{1}}(W) \setminus S^*(W)|+ 2|S_X|$.
Thus, $|U \cup B_1 \cup B_2| \geq 2(|S_{\RNum{1}}(W) \setminus S^*(W)| + |S_X|)/K$.

Let $F_{\RNum{1}} = U \cup B_1 \cup B_2 \cup S^*(W)$, which is the set of drivers in $S^*$ required to serve all of $\sigma_{\RNum{1}}(S_{\RNum{1}})$ in $(S^*, \sigma^*)$. Then $|F_{\RNum{1}}| = |U \cup B_1 \cup B_2| + |S^*(W)| \geq 2(|S_{\RNum{1}}(W) \setminus S^*(W)| + |S_X|)/K + |S^*(W)|$.
Recall that $S_{\RNum{1}} = S_{\RNum{1}}(W) \cup U \cup S_X$.
The ratio between the number of drivers in $S$ to serve $\sigma_{\RNum{1}}(S_{\RNum{1}}) \cup B_1 \cup B_2$ and $F_{\RNum{1}}$ is
\begin{align*}
\tag{7.1} \label{eq:1}
\frac{|S_{\RNum{1}} \cup B_1 \cup B_2|} {|F_{\RNum{1}}|}
\leq \frac{|S_{\RNum{1}} \cup F_{\RNum{1}}|} {|F_{\RNum{1}}|}
&\leq \frac{|S_{\RNum{1}}(W) \setminus S^*(W)|+|S_X|+|U \cup B_1 \cup B_2 \cup S^*(W)|} {|U \cup B_1 \cup B_2 \cup S^*(W)|} \\
&\leq \frac{|S_{\RNum{1}}(W) \setminus S^*(W)| + |S_X|} {2(|S_{\RNum{1}}(W) \setminus S^*(W)| + |S_X|)/K + |S^*(W)|} + 1\\
&\leq \frac{|S_{\RNum{1}}(W) \setminus S^*(W)| + |S_X|} {2(|S_{\RNum{1}}(W) \setminus S^*(W)| + |S_X|) / K} + 1\\
&= \frac{K}{2} + 1 = \frac{K + 2}{2}.
\end{align*}
Hence, it requires at least $|F_{\RNum{1}}| \geq 2|S_{\RNum{1}} \cup F_{\RNum{1}}| / (K+2)$ drivers in $S^*$ to serve all trips of $\sigma_{\RNum{1}}(S_{\RNum{1}})$.
\end{proof}

Notice that Equation~(\ref{eq:1}) holds when each driver in $u \in F_{\RNum{1}}$ serves at most $K$ trips of $\sigma_{\RNum{1}}(S_{\RNum{1}})$, that is, $|\sigma^*(u)|\leq K+1$.
Next, we consider the minimum number of drivers in $S^*$ that is required to serve all trips of $\sigma_{\RNum{2}}(S_{\RNum{2}})$ in $(S^*, \sigma^*)$,
where $(S_{\RNum{2}}, \sigma_{\RNum{2}})$ is the partial solution computed at the end of Phase-II.
Recall that $Z = \{i \in R \setminus \sigma_{\RNum{1}}(S_{\RNum{1}}) \mid \delta_i = 0\}$ and all of $Z$ are served in $\sigma_{\RNum{2}}(S_{\RNum{2}})$.

\begin{lemma}
Let $(S^*, \sigma^*)$ be any optimal solution for $(G,R)$.
Let $F_{\RNum{2}}$ be the set of drivers in $S^*$ that serve all trips of $\sigma_{\RNum{2}}(S_{\RNum{2}})$ in $(S^*, \sigma^*)$.
Then, $|F_{\RNum{2}}| \geq \frac{2|S_{\RNum{2}} \cup F_{\RNum{2}}|}{K+2}$.
\label{lemma-phase-two-min}
\end{lemma}

\begin{proof}
We consider $F_{\RNum{2}} = F_{\RNum{1}} \cup C'\cup V'\cup C''$, each of $C', V'$ and $C''$ is a set of drivers in $S^*$ that serves a portion of trips of $\sigma_{\RNum{2}}(S_{\RNum{2}}\setminus S_{\RNum{1}})$ in $(S^*,\sigma^*)$.
Let $S' = \{x \in S_{\RNum{2}} \setminus S_{\RNum{1}} \mid |\sigma_{\RNum{2}}(x)| = 1\}$ be the set of solo drivers in $S_{\RNum{2}} \setminus S_{\RNum{1}}$.
Since $S' \subseteq X$, $n_x > 0$ for every $x \in S'$.
Each $x \in S'$ belongs to a distinct node of $\Gamma$ since otherwise, one of them can serve the other.
This implies that $S' \subseteq Z$.
Let $C' = C'_0 \cup C'_1$ be the set of drivers in $S^*$ that serve all of $S'$ in $(S^*,\sigma^*)$, where $C'_0=\{v \in S^* \mid \sigma^*(v) \cap S' \neq \emptyset \text{ and } \delta_v = 0\}$ and $C'_1=\{v \in S^* \mid \sigma^*(v) \cap S' \neq \emptyset \text{ and } \delta_v \geq 1\}$.
By definition and $S' \subseteq Z$, $C' \subseteq X$ and $C'_1 \cap Z = \emptyset$.
Let $S'= S'_0 \cup S'_1$, where $S'_0$ is served by $C'_0$ and $S'_1$ is served by $C'_1$. Then $|C'_0|=|S'_0|$ because each $x \in S'_0$ belongs to a distinct node and $\delta_v = 0$ for all $v \in C'_0$.

Consider any driver $z \in S'_1$. Let $(S_z, \sigma_z)$ be the partial solution just before $z$ is assigned as a driver by the algorithm.
All trips in $C'_1 \cap R(A^*_{\node(z)})$ must have been assigned as drivers in $S_z$. Otherwise, any $v \in C'_1 \cap R(A^*_{\node(z)})$ would have been assigned as a driver in $S_{\RNum{2}}$ to serve $z$ when $\node(z)$ is processed.
Hence, $C'_1 \subseteq S_{\RNum{2}}$, and 
for every driver $v \in C'_1 \cap R(A^*_{\node(z)})$, $\free_z(v) = 0$ or $\free_z(v) > 0$ with $\stopp_z(v) = \delta_v$.
From these and each $z \in S'_1$ belongs to a distinct node, $|\bigcup_{v \in C'_1} \sigma_z(v) \setminus \{v\}| \geq |\bigcup_{v \in C'_1} \sigma^*(v) \cap S'_1| = |S'_1|$, and $|C'_1| \geq |S'_1|/K$ to serve all of $S'_1 \subseteq Z$ since $S'_1 \cap C'_1 = \emptyset$.
Recall that for every driver $v \in C'_1$, each passenger served by $\sigma_{\RNum{2}}(v)$ is either in $W$ or $Z$.
For any $v \in C'_1$ such that $\sigma_{\RNum{2}}(v) \cap W \neq \emptyset$,
$v \in S_{\RNum{1}}$ and $v$ is included in the calculation of Equation~(\ref{eq:1}).
For any such $v$ (regardless if $\sigma_{\RNum{2}}(v) \cap Z \neq \emptyset$), the ratio $|S_{\RNum{2}}\setminus S_{\RNum{1}}|/|F_{\RNum{2}}|$ decreases because $v \in S_{\RNum{1}}$ and $v \in F_{\RNum{2}}$.
To get the approximation ratio for the worst case, we assume that all $C'_1 \subseteq (S_{\RNum{2}} \setminus S_{\RNum{1}})$, that is, $\sigma_{\RNum{2}}(v) \cap W = \emptyset$ and $\sigma_{\RNum{2}}(v) \cap Z \neq \emptyset$ for all $v \in C'_1$.
Let $V'$ be the set of drivers in $S^*$ that serve all of $\bigcup_{v \in C'_1} \sigma_{\RNum{2}}(v) \cap Z$ in $(S^*,\sigma^*)$.
By the algorithm (Phase-II part 2 specifically), each passenger $z \in \sigma_{\RNum{2}}(v) \cap Z$ belongs to a distinct node, implying $|\bigcup_{v \in C'_1} \sigma_{\RNum{2}}(v) \cap Z| \geq |S'_1|$.
From these, each driver in $V' \subseteq S^*$ can serve at most $K$ trips of $\bigcup_{v \in C'_1} \sigma_{\RNum{2}}(v) \cap Z$, and hence, $|V'| \geq |S'_1|/K$. Since $S'_1\cap (\bigcup_{v\in C'_1} \sigma_{\RNum{2}}(v) \cap Z)=\emptyset$, $|V' \cup C'_1| \geq 2|S'_1|/K$.
In the worst case, the algorithm can assign all of $C'$ and $V'$ to be drivers in $S$.
From $|S'_0| = |C'_0|$, the ratio between $|S' \cup C' \cup V'|$ and $|C' \cup V'|$ is
\begin{align*}
\frac{|S' \cup C' \cup V'|}{|C'\cup V'|} \leq \frac{|S'| + |C' \cup V'|}{|C' \cup V'|} = \frac{|S'_0| + |S'_1|}{|C'_0|+|C'_1 \cup V'|} + 1 = \frac{|C'_0| + |S'_1|}{|C'_0|+ |C'_1 \cup V'|} + 1.
\end{align*}
Since $|S'_1| \geq |C'_1|$, $(|C'_0|+|S'_1|)/(|C'_0|+|C'_1 \cup V'|)$ is monotone decreasing in $|C'_0|$.
Therefore,
\begin{align*}
\tag{7.2} \label{eq:2}
\frac{|S' \cup C' \cup V'|}{|C'\cup V'|}
\leq \frac{|C'_0| + |S'_1|}{|C'_0|+|C'_1 \cup V'|} + 1
\leq \frac{|S'_1|}{|C'_1 \cup V'|} + 1
\leq \frac{|S'_1|}{2|S'_1|/K}+1
= \frac{K+2}{2}.
\end{align*}

Consider the remaining drivers in $S'' = S_{\RNum{2}} \setminus (S_{\RNum{1}} \cup S' \cup C')$.
Since each driver $x \in S''$ serves at least one passenger, $|\sigma_{\RNum{2}}(S'')| \geq 2|S''|$.
Let $C'' = C''_0 \cup C''_1$ be the set of drivers in $S^*$ that serve all of $\sigma_{\RNum{2}}(S'')$ in $(S^*,\sigma^*)$, where
$C''_0 = \{v \in S^* \mid \sigma^*(v) \cap \sigma_{\RNum{2}}(S'') \neq \emptyset \text{ and } v \in Z\}$ and $C''_1 = \{v \in S^* \mid \sigma_{\RNum{2}}(S'') \neq \emptyset \text{ and } v \in X \setminus Z\}$.
Note that $C''_0 \subseteq \sigma_{\RNum{2}}(S_{\RNum{2}})$ by definition.
From the algorithm (Phase-II), $\sigma_{\RNum{2}}(S'') \subseteq Z$ and $\sigma_{\RNum{2}}(S'') \cap \sigma_{\RNum{2}}(S') = \emptyset$.
In the worst case, each trip $v \in C''$ can be assigned as a driver in $S$ such that $v$ serves itself only.
From this, $C''_0 \subseteq \sigma_{\RNum{2}}(S_{\RNum{2}})$ and that every $v \in C''$ can serve at most $K$ trips of $\sigma_{\RNum{2}}(S'')$, the ratio between the number of drivers in $S$ that serve $\sigma_{\RNum{2}}(S'') \cup C''$ and $C''$ is
\begin{align*}
\tag{7.3} \label{eq:3}
\frac{|S'' \cup C''|}{|C''|} \leq \frac{|S''| + |C''|}{|C''|} \leq \frac{|S''|}  {2|S''|/K} + 1 \leq \frac{K}{2} + 1
= \frac{K + 2}{2}.
\end{align*}

Next, we combine Equations (\ref{eq:2}) and (\ref{eq:3}) with Equation (\ref{eq:1}).
Let $F_{\RNum{2}} = F_{\RNum{1}} \cup C' \cup V' \cup C''$, which is the set of drivers in $S^*$ required to serve all trips of $\sigma_{\RNum{2}}(S_{\RNum{2}})$ in $(S^*,\sigma^*)$.
Note that $F_{\RNum{1}} \subseteq S^*$ is the minimum set of drivers that serve all of $\sigma_{\RNum{1}}(S_{\RNum{1}}) \subseteq W$ and $(C' \cup V') \subseteq S^*$ is the minimum set of drivers that serve all of $S' \cup (\bigcup_{v \in S'} \sigma_{\RNum{2}}(v) \cap Z) $, and $C'' \subseteq S^*$ is the minimum set of drivers that serve all of $\sigma_{\RNum{2}}(S'') \subseteq Z$ such that $\sigma_{\RNum{2}}(S'') \cap \sigma_{\RNum{2}}(S') = \emptyset$.
The minimum number of drivers in each set of $F_{\RNum{1}}$, $C' \cup V'$ and $C''$ is calculated based on each driver $u \in F_{\RNum{2}}$ serving $K$ trips of $\sigma_{\RNum{2}}(S_{\RNum{2}})$, as stated in Equations (\ref{eq:1}), (\ref{eq:2}) and (\ref{eq:3}).
Hence, 
\begin{align*}
|F_{\RNum{2}}| &\geq 2|S_{\RNum{1}} \cup F_{\RNum{1}}|/(K+2) + 2|S' \cup C' \cup V'|/(K+2) + 2|S'' \cup C''|/(K+2) \\
&= 2(|S_{\RNum{1}} \cup F_{\RNum{1}}|+|S' \cup C' \cup V'|+|S'' \cup C''|)/(K+2).
\end{align*}
The ratio between $S_{\RNum{2}} \cup F_{\RNum{2}}$ and $F_{\RNum{2}}$ is at most
\begin{align*}
\tag{7.4} \label{eq:4}
\frac{|S_{\RNum{2}} \cup F_{\RNum{2}}|} {|F_{\RNum{2}}|}
&\leq \frac{|S_{\RNum{1}} \cup S' \cup S'' \cup F_{\RNum{1}} \cup C' \cup V' \cup C''|} {|F_{\RNum{1}} \cup C' \cup V' \cup C''|} \\
&\leq \frac{|S_{\RNum{1}} \cup F_{\RNum{1}}| +|S' \cup C' \cup V'| + |S'' \cup C''|} {2(|S_{\RNum{1}} \cup F_{\RNum{1}}|+|S' \cup C' \cup V'|+|S'' \cup C''|)/(K+2)} \\
&= \frac{K+2}{2}.
\end{align*}
Therefore, it requires at least $|F_{\RNum{2}}| \geq 2|S_{\RNum{2}} \cup F_{\RNum{2}}|/(K+2) \geq 2|S_{\RNum{2}}| / (K+2)$ drivers in $S^*$ to serve all of $\sigma_{\RNum{2}}(S_{\RNum{2}})$.
\end{proof}

Again, Equation (\ref{eq:4}) holds if each driver $u \in F_{\RNum{2}}$ serves at most $K$ trips of $\sigma_{\RNum{2}}(S_{\RNum{2}})$, that is, $|\sigma^*(u)|\leq K+1$.
Recall that $B_1$ and $B_2$ are subsets of $F_{\RNum{1}}$ defined in the proof of Lemma~\ref{lemma-phase-one-min}, and $C'$, $V'$ and $C''$ are subsets of $F_{\RNum{2}}$ defined in the proof of Lemma~\ref{lemma-phase-two-min}.
Each trip $v$ in $B_1 \cup B_2 \cup C' \cup V' \cup C''$ can be a driver in $S$ that serves itself only.
This can happen if before $v \in B_1 \cup B_2 \cup C' \cup V' \cup C''$ is processed by the algorithm, $R(D^*_{\node(v)}) \setminus \{v\} \subseteq \sigma_v(S_v)$ for $\delta_v > 0$ or $R(\node(v)) \setminus \{v\} \subseteq \sigma_v(S_v)$ for $\delta_v = 0$, where $(S_v, \sigma_v)$ is the partial solution just before $v$ is processed.
In other words, all trips that can be served by $v$ are already served w.r.t. $(S_v, \sigma_v)$.

\begin{remark}
Let $B_1$ and $B_2$ be the set of trips defined in the proof of Lemma~\ref{lemma-phase-one-min}. Let $S'$ and $S''$ be the sets of drivers defined in the proof of Lemma~\ref{lemma-phase-two-min}.
Trips of $B_1 \cup B_2$ can be assigned as drivers in Phase-II or Phase-III. Suppose $v \in B_1 \cup B_2$ is assigned as a driver in Phase-II.
If $\sigma(v)$ serves only itself, $v$ is included in $S'$.
If $\sigma(v)$ serves more than one trip, $v$ is included in $S''$.
For either case, Equation (\ref{eq:4}) holds.
\label{remark-1}
\end{remark}

From Remark~\ref{remark-1}, let $B'_1 \cup B'_2 \subseteq B_1 \cup B_2$ be the trips assigned as drivers in Phase-III.
Let $\bar{S} = S \setminus (S_{\RNum{2}} \cup B'_1 \cup B'_2 \cup C' \cup V' \cup C'')$ be the set of drivers found during Phase-III of the algorithm.

\begin{lemma}
It requires at least $\frac{2|S|}{K+2}$ drivers in $S^*$ to serve all trips of $\sigma(S)$ in $(S^*, \sigma^*)$.
\label{lemma-phase-three-min}
\end{lemma}

\begin{proof}
Any trip $x$ in $R \setminus \sigma_{\RNum{2}}(S_{\RNum{2}})$ has $n_x > 0$ and $\delta_x > 0$ since all of $W$ and $Z$ are served in $\sigma_{\RNum{2}}(S_{\RNum{2}})$.
Consider the moment a trip $x \in \bar{S}$ is assigned as a driver.
Let $(S_x, \sigma_x)$ be the partial solution just before $x$ is processed by the algorithm.
Since $n_x > 0$ and $\delta_x > 0$, $x$ will serve at least one passenger ($|\sigma(x)| \geq 2$) if there exists an un-assigned trip in $R(D^*_{\node(x)}) \setminus \{x\}$, that is, $R(D^*_{\node(x)}) \setminus \{x\} \nsubseteq \sigma_x(S_x)$.
Let $X(1) = \{x \in \bar{S} \mid |\sigma(x)| = 1\}$.
For every pair $x, x' \in X(1)$, $x \notin R(D^*_{\node(x')}) \cup R(A^*_{\node(x')})$. Otherwise, one of them can serve the other.
For every $x \in X(1)$, any driver $x' \in \bar{S}(A^*_{\node(x)}) \setminus \{x\}$ must serve at least two trips, where $\bar{S}(A^*_{\node(x)}) = \bar{S} \cap R(A^*_{\node(x)})$.
For any $x \in X(1)$, let $Y_x$ be the set of drivers in $S^*$ that serve all of $\bigcup_{x' \in \bar{S}(A^*_{\node(x)})} \sigma(x')$ in $(S^*,\sigma^*)$.
For a driver $y \in Y_x$, $\sigma^*(y) \setminus \{y\}$ can contain at most $K$ trips of $\bigcup_{x' \in \bar{S}(A^*_{\node(x)})} \sigma(x')$. If $y \in \bigcup_{x' \in \bar{S}(A^*_{\node(x)})} \sigma(x')$, $y$ can serve at most $K+1$ trips of $\bigcup_{x' \in \bar{S}(A^*_{\node(x)})} \sigma(x')$.
Hence,
$|Y_x| \geq (2|\bar{S}(A_{\node(x)})| + 1)/(K+1)$.
For any pair $x, x' \in X(1)$ with $x \neq x'$, since drivers in $Y_x$ cannot serve any trip of $\bigcup_{x'' \in \bar{S}(A^*_{\node(x')})} \sigma(x'')$, it must be that $Y_x \cap Y_{x'} = \emptyset$.
Let $Y = \bigcup_{x \in X(1)} Y_x$, $\bar{S}_{X(1)} = \bigcup_{x \in X(1)} \bar{S}(A^*_{\node(x)})$ and $\sigma(\bar{S}_{X(1)}) = \bigcup_{x \in \bar{S}_{X(1)}} \sigma(x)$.
Note that $|Y| \geq |X(1)|$.
Then, 
\begin{align*}
|Y| = \sum_{x \in X(1)} |Y_x| &\geq \sum_{x \in X(1)} (2|\bar{S}(A_{\node(x)})| + 1)/(K+1) \\
&= (\sum_{x \in X(1)} 2|\bar{S}(A_{\node(x)})| + |X(1)|)/(K+1).
\end{align*}
The above can be rewritten as 
\[
\frac{(K+1)\cdot |Y| - |X(1)|}{2} \geq \sum_{x \in X(1)} |\bar{S}(A_{\node(x)})| = \sum_{x \in X(1)} |\bar{S}(A^*_{\node(x)})| - X(1),
\]
and hence,
\[
\sum_{x \in X(1)} |\bar{S}(A^*_{\node(x)})| \leq \frac{(K+1)\cdot |Y| + |X(1)|}{2}.
\]

Consider the remaining drivers in $\bar{S}_{\overline{X(1)}} = \bar{S} \setminus \bar{S}_{X(1)}$. Each driver $x \in \bar{S}_{\overline{X(1)}}$ serves at least two trips, implying $|\sigma(\bar{S}_{\overline{X(1)}})| \geq 2|\bar{S}_{\overline{X(1)}}|$.
Let $Y'$ be the set of drivers in $S^*$ that serve all of $\bigcup_{x \in \bar{S}_{\overline{X(1)}}} \sigma(x)$ in $(S^*,\sigma^*)$.
Any driver $x \in \bar{S}_{\overline{X(1)}}$ is not in $\sigma(\bar{S}_{X(1)})$ by the definition of $\bar{S}_{\overline{X(1)}}$, and $x$ is not in $R(D^*_{\node(x')})$ for any $x' \in X(1)$ since otherwise, $x'$ would have served $x$.
From these, for every $y' \in Y'$, $y' \notin R(D^*_{\node(x)}) \cup R(A^*_{\node(x)})$ for all $x \in \bar{S}_{X(1)}$, which implies that $|Y \cup Y'| = |Y| + |Y'|$.
Similar to $Y$, each driver in $Y'$ can serve at most $K+1$ trips of $\bigcup_{x \in \bar{S}_{\overline{X(1)}}} \sigma(x)$.
Hence, $|Y'| \geq 2|\bar{S}_{\overline{X(1)}}|/(K+1)$, implying $((K+1)\cdot |Y'|)/2 \geq |\bar{S}_{\overline{X(1)}}|$.
Each $y \in Y \cup Y'$ must be in either $\sigma(S\setminus S_{\RNum{2}})$ or $\sigma(S_{\RNum{2}})$ since all trips must be served at the end by the algorithm.
In other words, if $y \in S \setminus \bar{S}$, $y$ has been considered in Equation (\ref{eq:4}). This means that we only need to consider $\bar{S}$, and the ratio between $|\bar{S}|$ and $|Y \cup Y'|$ is
\begin{align*}
\tag{7.5} \label{eq:5}
\frac{|\bar{S}|}{|Y \cup Y'|} = \frac{|\bar{S}_{X(1)}| + |\bar{S}_{\overline{X(1)}}|} {|Y| + |Y'|}
&\leq \frac{((K+1)\cdot |Y| + |X(1)|)/2 + ((K+1)\cdot |Y'|)/2} {|Y| + |Y'|} \\
&= \frac{(K+1)\cdot (|Y| + |Y'|) + |X(1)|} {2(|Y| + |Y'|)} \\
&= \frac{K+1}{2} + \frac{|X(1)|} {2(|Y| + |Y'|)} \\
&\leq \frac{K+1}{2} + \frac{1}{2} =\frac{K+2}{2}.
\end{align*}

Finally, we calculate the ratio between $S$ and $S^*$, where $F_{\RNum{2}} \cup Y \cup Y' \subseteq S^*$.
Recall that $S = \bar{S} \cup S_{\RNum{2}} \cup F_{\RNum{2}}$ and all trips served by each driver $x \in \bar{S}$ are in $X$.
From this and the same reason stated in the proof of Lemma~\ref{lemma-phase-three-min} for Equation (\ref{eq:4}),
the minimum number of drivers in each set of $F_{\RNum{2}}$, $Y$ and $Y'$ is calculated based on using all capacity $K$ of every driver $u \in F_{\RNum{2}} \cup Y \cup Y'$.
Hence, with Equations (\ref{eq:4}) and (\ref{eq:5}),
\begin{align*}
|F_{\RNum{2}} \cup Y \cup Y'| &\geq 2|S_{\RNum{2}} \cup F_{\RNum{2}}|/(K+2) + 2|\bar{S}|/(K+2) \\
&= 2(|S_{\RNum{2}} \cup F_{\RNum{2}}| + |\bar{S}|)/(K+2)
\end{align*}
The ratio between $S$ and $S^*$ is
\begin{align*}
\tag{7.6} \label{eq:6}
\frac{|S|} {|S^*|}
\leq \frac{|S|} {|F_{\RNum{2}} \cup Y \cup Y'|}
&\leq \frac{|\bar{S}| + |S_{\RNum{2}} \cup F_{\RNum{2}}|} {|F_{\RNum{2}} \cup Y \cup Y'|} \\
&\leq \frac{|\bar{S}| + |S_{\RNum{2}} \cup F_{\RNum{2}}|} {2(|S_{\RNum{2}} \cup F_{\RNum{2}}| + |\bar{S}|)/(K+2)} \\
&= \frac{K+2}{2}
\end{align*}
Therefore, it requires at least $\frac{2|S|}{K+2}$ drivers in $S^*$ to serve all of $\sigma(S)$.
\end{proof}

Next, we show that Algorithm~\arabic{algcounter} always computes a valid solution to any instance of the ridesharing problem with stop constraint, followed by its running time.
Let $(S', \sigma')$ be the partial solution computed by Algorithm~\arabic{algcounter} for a given time point.
\begin{lemma}
Let $(S, \sigma)$ be a solution found by Algorithm~\arabic{algcounter} after processing all trips in $R$.
Then for each pair $i, j \in S$, $\sigma(i) \cap \sigma(j) = \emptyset$ and $\sigma(S) = R$, implying $(S, \sigma)$ is indeed a valid solution to the ridesharing instance $(G, R)$.
\label{lemma-alg-solution}
\end{lemma}

\begin{proof}
Phase-I of the algorithm ends until all trips of $W$ are served, that is, $\Gamma(W) = \emptyset$.
In each iteration of Phase-I, a node $\mu \in \Gamma(W)$ containing trips of $W$ is chosen w.r.t. $(S', \sigma')$.
A trip $x$ is selected from $\hat{X}_1 \cup \bar{X}$, where $\hat{X}_1 = \{i \in S'(A_{\mu}) \mid \free'(i) > 0 \text{ and } \stopp'(i) < \delta_i \}$ and $\bar{X} = \{i \in X \cap R(A^*_{\mu}) \setminus \sigma'(S') \mid \stopp'(i) < \delta_i \text{ or } i \in R(\mu)\}$.
By the definition of $\hat{X}_1$ and $\bar{X}$, $x$ is either a driver or an un-assigned trip that can still serve other trips in $R(D^*_{\mu})$.
From this, $x$ is a valid assignment for serving $W(x, \mu, S')$.
If $\hat{X}_1 \cup \bar{X} = \emptyset$, each trip of $W(\mu) \setminus \sigma'(S')$ is assigned as a driver to serve itself.

Phase-II of the algorithm ends until all trips of $Z$ are served, where $Z = \{i \in R \setminus \sigma'(S') \mid \delta_i = 0\}$. 
Since all of $W$ are served before Phase-II starts, $n_i \geq 1$ for every $i \in R \setminus \sigma'(S')$, that is, $Z \subseteq X$.
From this, every $x \in Z \setminus \sigma'(S')$ that is assigned as a driver to serve other trips in $Z(x, \node(x), S')$ is valid, as described in the first part (first for-loop) of Phase-II.
The second part of Phase-II is similar to Phase-I.
A node $\mu \in \Gamma(Z)$ is chosen, where $\Gamma(Z)$ is the set of nodes containing the rest of $Z$ w.r.t. $(S', \sigma')$.
Either a driver $x \in \hat{X}_2$ or an unassigned trip $x \in \bar{X}$ is selected to serve $Z(x, \mu, S')$, where $\hat{X}_2 = \{i \in S(A^*_{\mu}) \mid \free(i) > 0 \text{ and } (\stopp(i) < \delta_i \text{ or } i \in R(\mu))\}$ and $\bar{X}$ is the same as defined above.
The assignment of $x$ is valid as mentioned above.

In Phase-III of the algorithm, the rest of $X \setminus Z$ are served.
The algorithm processes each node from $\mu_p$ to $\mu_1$.
All trips in $R(\mu_j)$ must be served before $\mu_{j-1}$ is processed.
In each iteration, either a driver $x \in \hat{X}_2$ (as defined above) or an unassigned trip $x \in X(\mu) \setminus \sigma'(S')$ is selected to serve $R(x, \mu, S')$.
The assignment of $x$ is valid as mentioned above.
Therefore, Algorithm~\arabic{algcounter} produces a valid solution after all trips in $R$ are processed.
\end{proof}

\begin{theorem}
Given a ridesharing instance $(G, R)$ of size $M$ and $l$ trips satisfying Conditions (1-3) and (5).
Algorithm~\arabic{algcounter} computes a solution $(S, \sigma)$ for $(G, R)$ such that $|S^*| \leq |S| \leq \frac{K+2}{2} |S^*|$, where $(S^*, \sigma^*)$ is any optimal solution and $K = \max_{i \in R} {n_i}$, with running time $O(M + l^2)$.
\end{theorem}

\begin{proof}
By Lemma~\ref{lemma-phase-three-min} and Lemma~\ref{lemma-alg-solution}, Algorithm~\arabic{algcounter} computes a solution $(S, \sigma)$ for $R$ with $\frac{K+2}{2}$-approximation ratio.
It takes $O(M)$ time to construct the meta graph $\Gamma(V, E)$ using the preprocessing described in~\cite{GLZ19}.
The labeling of nodes in $\Gamma$ takes $O(l)$ time.
Sorting the trips in a node $\mu$ according to their capacity takes $O(K \cdot |R(\mu)|)$ time for each node $\mu$, so in total $O(K \cdot l)$ to sort all trips in $R$.The total time for the preprocessing is $O(M+K\cdot l)$; we assume $K < l$.
For Phase-I, there are at most $O(l)$ iterations (in the while-loop).
In each iteration, it takes $O(l)$ time to pick the required node $\mu$ from $\Gamma(W)$ and $O(l)$ time to select a trip $x$ from $\hat{X}_1 \cup \bar{X}$.
To serve all of $W(x, \mu, S')$ or $W(\mu)$, $|W(\mu)| \leq O(l)$ is required.
Hence, Phase-I runs in time $O(l^2)$.
For Phase-II, we can first scan the tree $\Gamma$ following the node labels in decreasing order, which takes $O(l)$ time.
Whenever a node $\mu$ with $|Z(\mu)| \geq 2$ is encountered, a trip $x \in Z(\mu) \setminus \sigma'(S')$ is selected to serve $Z(x, \mu, S')$ repeated until $|Z(\mu)| \leq 1$. This takes $O(l)$ time since the trips in $R(\mu)$ are sorted  according to their capacity.
Hence, it takes $O(l^2)$ time for the first for-loop in Phase-II.
The second for-loop in Phase-II is similar to Phase-I, which requires $O(l)$ time for each iteration.
Thus, it requires $O(l^2)$ time for Phase-II.
For Phase-III, in each iteration when processing a node $\mu$, it takes $O(l)$ time to select a trip $x$ from $\hat{X}_2$ or $X(\mu) \setminus \sigma'(S')$. Then in total, it requires $O(l + K)$ time to serve $R(x, \mu, S')$.
Collectively, Phase-III may require $O(l)$ iterations to process trips of all nodes in $V(\Gamma)$.
Thus, it requires $O(l^2)$ time for Phase-III.
Therefore, the running time of Algorithm~\arabic{algcounter} is $O(M + l^2)$. 
\end{proof}

\section{Conclusion} \label{conclusion}
We proved that it is NP-hard to approximate with a constant factor each problem of minimizing the number of drivers and minimizing the total travel distance of drivers if one of Conditions (2)-(5) is not satisfied. Our results together with the results in~\cite{GLZ16} imply that both minimization problems are NP-hard if one of Conditions (1)-(5) is not satisfied.
We also presented $\frac{K+2}{2}$-approximation algorithms for minimizing number of drivers for problem instances satisfying all conditions except Condition (4), where $K$ is the largest capacity of all vehicles.
It is worth developing approximation algorithms for other NP-hard cases; for example, two or more of the five conditions are not satisfied.
It is interesting to study applications of the approximation algorithms for other related problems, such as multimodal transportation with ridesharing (integrating public and private transportation).

\end{document}